\documentclass[letterpaper, 10 pt, conference]{ieeeconf}  % Comment this line out if you need a4paper

\usepackage{amsthm,amsmath,amssymb,mathtools,newtxmath}
\usepackage{graphicx}
\usepackage{cite}
\usepackage{cleveref}

\newtheorem{lemma}{Lemma}
\newtheorem{theorem}{Theorem}
\newtheorem{problem}{Problem}
\newtheorem{remark}{Remark}

\usepackage{booktabs,lipsum}
\usepackage{subcaption}
\usepackage{color}

\IEEEoverridecommandlockouts                              % This command is only needed if 
\title{\LARGE \bf
Self-Triggered Markov Decision Processes
}

\author{Yunhan Huang$^{1}$ and Quanyan Zhu$^{1}$% <-this % stops a space
\thanks{$^{1}$ Y. Huang and Q. Zhu are with the Department of Electrical and Computer Engineering,
        New York University, 370 Jay St., Brooklyn, NY.
        {\tt\small \{yh.huang, qz494\}@nyu.edu}}%%
}

\begin{document}

\maketitle
\thispagestyle{plain}
\pagestyle{plain}

%%%%%%%%%%%%%%%%%%%%%%%%%%%%%%%%%%%%%%%%%%%%%%%%%%%%%%%%%%%%%%%%%%%%%%%%%%%%%%%%
\begin{abstract}
In this paper, we study Markov Decision Processes (MDPs) with self-triggered strategies, where the idea of self-triggered control is extended to more generic MDP models. This extension broadens the application of self-triggering policies to a broader range of systems. We study the co-design problems of the control policy and the triggering policy to optimize two pre-specified cost criteria. The first cost criterion is introduced by incorporating a pre-specified update penalty into the traditional MDP cost criteria to reduce the use of communication resources. Under this criteria, a novel dynamic programming (DP) equation called DP equation with optimized lookahead to proposed to solve for the self-triggering policy under this criteria. The second self-triggering policy is to maximize the triggering time while still guaranteeing a pre-specified level of sub-optimality. Theoretical underpinnings are established for the computation and implementation of both policies. Through a gridworld numerical example, we illustrate the two policies' effectiveness in reducing sources consumption and demonstrate the trade-offs between resource consumption and system performance.
\end{abstract}

%%%%%%%%%%%%%%%%%%%%%%%%%%%%%%%%%%%%%%%%%%%%%%%%%%%%%%%%%%%%%%%%%%%%%%%%%%%%%%%%
\section{Introduction}
Recent advances in information and communication technologies have led to the implementation of large-scale resource-constrained networked control systems. In these systems, it is desirable to limit the sensor and control communication and computation to instances when a system needs attention \cite{heemels2012introduction}. As a result, the self-triggered control paradigm is proposed to reduce the utilization of communication resources and/or actuation movements while still maintaining desirable closed-loop behavior for these systems \cite{gommans2014self}. The self-triggered control abandons the conventional periodic time-triggered implementations. In self-triggered control, the self-triggering policy consists of two sub-policies: the control policy and a triggering mechanism that pre-determines, at an update time, when the control inputs have to be updated the next time. Due to its efficiency in resource-saving, self-triggered control has been studied extensively in the last decades \cite{anta2010sample,wang2009self,akashi2018self,heemels2012introduction,gao2019robust,gommans2014self}.

The study of self-triggered has been confined to state-space dynamical models, including either linear models \cite{heemels2012introduction,wang2009self,akashi2018self,gommans2014self} or nonlinear models \cite{anta2010sample,gao2019robust} in both(either) continuous-time and(or) discrete-time settings. However, recent developments in technologies such as wireless communication, machine learning, and real-time analytics have broadened the application of Internet of Things (IoTs) beyond control systems to a wide range of areas, including logistics and supply chain \cite{yuvaraj2016smart,feinberg2016optimality,wikecek2016mean}, smart cities\cite{zanella2014internet}, and wearables\cite{lu2019wearable}. These systems are usually large-scale, equipped with resource-constrained devices, and difficult to be described by state-space dynamic models. Hence, there is an urgent need to incorporate the idea of self-triggering policy in control into a more general dynamic model: Markov Decision Processes (MDP). This incorporation can lead toward a computationally and communicationally more efficient IoT-enabled system.

This paper studies a discrete-time self-triggered MDP where the control\footnote{In this paper, we use control and action interchangeably.} policy and the triggering mechanism/policy are co-designed to achieve certain cost criteria. The differences between this work and most existing papers in self-triggered control are three-fold. The first is that we study self-triggered policies for a more generic dynamic model, i.e., an MDP model, which allows the extension of the self-triggering policy to a wider range of applications. Second, we address the co-design problem of jointly designing the control policy and the triggering policy. Existing self-triggering methods design the control policy and the triggering policy in an ordered manner, i.e., the control policy is designed first. The triggering policy is then designed subsequently while ensuring certain control performance \cite{heemels2012introduction,wang2009self}. For example, in \cite{wang2009self}, the control gain is pre-set to be the $H_\infty$ control gain, based on which a triggering policy is designed to assure a specified level of $\mathcal{L}_2$ stability. Since the control policy is given without considering the self-triggering nature of the whole policy, it is hard to guarantee that the given control policy is optimal for achieving the minimum number of updates while maintaining certain cost criteria \cite{gommans2014self}. Here, we address a co-design problem to alleviate the concern regarding the optimality issue. Third, in existing works \cite{heemels2012introduction}, the analysis of control performance under the self-triggered control paradigm is mostly qualitative, e.g.,  the analysis of whether a certain type of stability can be achieved.  Control performance is sometimes quantified as the decay rate for the Lyapunov function. Only few self-triggering methods provide quantitative analysis for control performance such as $\mathcal{L}_2$ gains \cite{wang2009self}, quadratic costs\cite{molin2013optimality,maity2019optimal,huang2020infinite}. More recently, T. Gommans et al. studies self-triggered linear-quadratic-gaussian (LQG) control associated with quadratic costs. In this work, we consider a generic class of cost criteria and propose self-triggered policies that can guarantee a certain optimality level.The contributions of this paper are summarized as follows.
\begin{enumerate}
    \item We study self-triggered MDP, which extends the idea of self-triggered control into a more generic dynamical model. The genericness of the MDP model enables the application of self-triggering policies into a broader range of systems.
    \item We jointly design the control policy and the triggering policy that co-optimizes pre-specified cost criteria.
    \item We propose two frameworks that produce two co-designed self-triggering policies. The first is introduced by incorporating an update penalty into the traditional MDP cost criteria to reduce the use of communication resources. The second is a greedy reduction of resources used while still guaranteeing any pre-given level of sub-optimality. Theoretical underpinnings are established for the computation and implementation of both policies.
    \item Through a gridworld example in both non-windy and windy settings, we show that the proposed policies are efficient in reducing communication resources consumed while still maintaining a high level of performance.
\end{enumerate}

%For the rest of the paper, \Cref{sec:ProblemForm} provides a mathematical formulation of the problem. \Cref{Sec:TheoreticalResults} presents the theoretical results of the framework, and their derivations. In \Cref{sec:NumericalStudy}, we use case studies to illustrate the structures of the observation and jamming strategies.

\subsection{Nomenclature}

In this paper, $\mathbb{R}$ and $\mathbb{N}$ represent the set of real numbers and natural numbers, respectively. The expectation operator is denoted by $\mathbb{E}$. And $\Delta t \in \mathbb{N}$ denotes the time steps between two neighboring updates. The letter $l$ is the index for the $l$th update and $t_l$ is the time instance when the $l$th update happens. The notation $\mathbb{N}_{[t_l,t_{l+1}]}$ means the intersection of the two sets$\mathbb{N}$ and $[t_l,t_{l+1}]$. The set of non-negative real numbers is denoted by $\mathbb{R}^+$. The notation $\mathcal{A}\backslash \mathcal{B}$ denotes the set $\{x\ |\ x\in\mathcal{A},x\notin\mathcal{B}\}$.

\section{SELF-TRIGGERED MARKOV DECISION PROCESS}\label{sec:ProblemForm}

In this section, we provide the problem formulation for the self-triggered action strategy. We consider a discrete-time MDP defined by a tuple $\{\mathcal{X},\mathcal{A},P,c\}$, where $\mathcal{X}$ is the state space, $\mathcal{A}$ is the actions space, $P$ is the time-homogeneous transition probability, and $c$ is the state-wise cost function. The state space $\mathcal{X}$ and action space $\mathcal{A}$ are both assumed to be Borel subsets of Polish (Banach and separable) spaces. If an action $a\in \mathcal{A}$ is selected at a state $x\in\mathcal{X}$, then a cost $c(x,a)$ is incurred, where without loss of generality, we suppose $c:\mathcal{X}\times \mathcal{A}\rightarrow\mathbb{R}^+$. The function $c$ is assumed to be bounded and Borel measurable. The transition probability $P(B|x,a)$ is a Borel function on $\mathcal{X}\times \mathcal{A}$ for each Borel subset $B$ of $\mathcal{X}$, and $P(\cdot|x,a)$ is a probability measure on the Borel $\sigma$-field of $\mathcal{X}$ for each $(x,a)\in\mathcal{X}\times \mathcal{A}$.

In classic MDP, the decision process proceeds as follows: at time $t=0,1,\cdots$, the current state of the system, $x_t$, is observed. A decision-maker decides which action, $a_t$, to choose, the cost $c_t = c(x_t,a_t)$ is incurred, the system moves to the next state following the rule $x_{t+1} \sim P(\cdot|x_t,a)$, and the process continues. The rule that the decision-maker follows to choose an action is called policy. We consider stationary Markov policy $\phi$ in which all decisions depend only on the current state. A stationary Markov policy $\phi$ is defined by a measurable mapping $\phi:\mathcal{X}\times \mathcal{A}$. In classic MDP, the goal is to find an optimal stationary Markov policy that minimizes 
\begin{equation}\label{ClassicCostFunctional}
v^\phi(x) = \mathbb{E}^{\phi}\left[\sum_{t=0}^\infty \beta^t c(x_t,a_t)\middle \vert x_0=x \right],
\end{equation}
where $\beta$ is a discount factor strictly less than $1$, and the expectation is based on the probability distribution on the set of all trajectory $(\mathcal{X}\times \mathcal{A})^\infty$, which is uniquely determined by the policy $\phi$ and the initial state $x$ (\cite{bertsekas1996stochastic}, pp. 140-141). Define the optimal cost
$$
V(x)\coloneqq \inf_{\phi \in \Phi} v^\phi(x),
$$
where $\Phi$ is the set of all stationary policies. A policy $\phi$ is called optimal if $v^\phi(x)=V(x)$ for all $x\in\mathcal{X}$.

\subsection{Self-Triggered Decision Making}
In classic MDP, decision making requires persistent transmission of measured state and updates of actions at each time instance $t\in\mathbb{N}$. In this paper, we are interested in constructing a policy that requires less sensing demand, lower communication rate, and less actuator movements  \cite{gallieri2012}, while still maintaining certain forms of optimality. 

The self-triggering policy is based on holding the current input value for a controlled duration while still guaranteeing certain forms of optimality. The self-triggered policy carries the following structure
\begin{equation}\label{Eq:StructureSTP}
\begin{cases}
t_{l+1} &= t_l + \tau(x_{t_l}),\\
a_t &= \pi(x_{{t_l}})\in\mathcal{A},\ t\in \mathbb{N}_{[t_l.t_{l+1})},
\end{cases}
\end{equation}
where $l$ is the index for the number of triggers,  $t_0\coloneqq 0$, $\tau:\mathcal{X}\rightarrow  \mathcal{T}, \mathcal{T}\coloneqq \{1,2,\cdots,\bar{T}\}$, $\bar{T}\in\mathbb{N}$, and $\pi:\mathcal{X}\rightarrow \mathcal{A}$. Here, the integer $\bar{T}$ is an arbitrary upper bound on the waiting time for next update. The self-triggering policy, denoted by $\mu$, involves two sub-policies: the timing policy, $\tau(x)$, that determines the next time for updating, and the control policy, $\pi(x)$, that chooses a fixed action to deploy for the next $\tau(x)$ time instances. For convenience, we write $\mu = (\tau, \pi)$ and $\mu:\mathcal{X}\rightarrow \mathcal{T}\times \mathcal{A}$.

\subsection{Performance Criteria}
This paper introduces two different yet related problems associated with two cost criteria; one is constructed by incorporating a penalty $O\geq 0$ for updating the action into the classic cost criteria defined in \cref{ClassicCostFunctional}. The idea of introducing a penalty is originated from costly measurements that have been investigated in the context of LQG optimal control \cite{huang2020infinite,cooper1971optimal} and games\cite{huang2020cross,huang2021pursuit}. The penalty $O\geq 0$ is a scalar, which we refer to as the update penalty. For instance, if $t_l$ and $t_{l+1}$ are two neighboring updating time, during the time interval $[t_{l},t_{l+1}]$, the total update penalty is $\beta^{t_l}O + \beta^{t_{l+1}}O$. Now, we formulate the first problem.
\begin{problem}\label{ProblemOfPenalty}
Find an optimal self-triggering policy $\mu$ that minimizes the following cost criterion over an infinite horizon
\begin{equation}\label{Eq:CostFunctionalPenalty}
f^\mu(x) = \mathbb{E}^{\mu}\left[ \sum_{t=0}^\infty \beta^t c(x_t,a_t) +\sum_{l=1}^{\infty}  \beta^{t_l} O \middle \vert x_0 =x \right],
\end{equation}
where the first term is the accumulated costs in the classic MDP, and the second term is the accumulated costs of updating one's action.
\end{problem}
The other cost criteria is similar to that of \cite{gommans2014self}. That is for a pre-specified sub-optimal performance, we aim to reduce the number of times the input/output is updated, while maintaining the pre-specified sub-optimal performance. Now, we formulate our second problem as
\begin{problem}\label{ProblemofHardConstraint}
Find a policy $\mu$ that maximizes the next transmission time $\tau(x)$ subject to the performance guarantee that 
\begin{equation}\label{Eq:Sub-optimalityLevel}
v^\mu(x) \leq \alpha V(x),\ \textrm{for all }x\in\mathcal{X},
\end{equation}
where $\alpha\geq 1$ is a scalar. 
\end{problem}
\begin{remark}
In Problem \ref{ProblemOfPenalty}, we introduce an update penalty $O$ to capture the trade-off between the degree of optimality and the usage of sensing/communication resources. The update penalty can be interpreted as a soft constraint on the number of updates. In Problem \ref{ProblemofHardConstraint}, $\alpha$ serves as a scaling factor that can be selected arbitrarily to balance the consumption of sensing/communication resources and the degradation of performance. There is a hard constraint that requires matainting a certain degree of sub-optimality. When $\alpha=1$, no degradation of performance is allowed. Solving both problems involves the co-design of the waiting time for next update (through $\tau$) and the chosen action (through $\pi$).
\end{remark}

\section{THEORETICAL FRAMEWORKS}\label{Sec:TheoreticalResults}

In this section, by establishing theoretical underpinnings, we pave the way for finding the self-triggering policies that solve the problems. For \Cref{ProblemOfPenalty}, we formulate a dynamic programming (DP) equation, which we call a DP equation with optimized lookahead. With this equation, we can resort to several effective methods such as value iterations and policy iterations to characterize an optimal self-triggering policy. For \Cref{ProblemofHardConstraint}, we propose a greedy self-triggering policy that aims to reduce the number of updates and show that the proposed policy is well-defined and satisfies the performance guarantee for any pre-specified $\alpha$.

\subsection{Dynamic Programming Equation with Optimized Lookahead}

To solve \Cref{ProblemOfPenalty}, the DP equation with optimized lookahead is derived and presented in this sub-section. The derivation idea is to form consolidated costs, states, and actions between two update time instances, which generates a new discrete-time MDP in the classic setting.

Let $\bar{c}_l$ represent the consolidated costs that correspond to the time period between $l$-th update and $(l+1)$-th update, i.e., the time period $[t_{l},t_{l+1})$. From \cref{Eq:StructureSTP} and \cref{Eq:CostFunctionalPenalty}, we can obtain
$$
\begin{aligned}
\bar{c}_l &\coloneqq \bar{c}(x_{t_l}, a_{t_l},\Delta t_l)
= \mathbb{E}\left[ \sum_{t=0}^{\Delta t_l-1} \beta^t c(x_{t_l + t},a_{t_l}) \middle \vert x_{t_l},a_{t_l},\Delta t_l  \right],
\end{aligned}
$$
where given a self-triggering policy $\mu=(\pi,\tau)$, the fixed action $a_{t_l}$ is produced by $\pi(x_{t_l})$ and the waiting time $\Delta t_l$ is generated by $\tau(x_{t_l})$. An application of the Fubini's theorem (principle)  and Markov property \cite{durrett2019probability} yields
$$
\bar{c}(x,a,\Delta t) = \sum_{t=0}^{\Delta{t}-1} \beta^t \mathbb{E}\left[ c(x_t,a_t) \middle \vert x_0 =x, a_t =a,\forall t<\Delta t \right].
$$
Furthermore, we define 
\begin{equation}\label{Eq:Skip-TransitionProb}
\bar{P}(B|x,a,\Delta T)\coloneqq \textit{Prob}\left(x_{\Delta t} \in B\middle\vert x_0=x, a_t= a,\forall t<\Delta t \right),
\end{equation}
as the skip-probability that the MDP is in Borel subset $B$ of $\mathcal{X}$, after time $\Delta t$, given that the initial condition is $x_0 =x$ and that the action is fixed until $\Delta t$. The skip-probability $\bar{P}(B|x,a,\Delta T)$ is a Borel function on $\mathcal{X}\times \mathcal{A}\times\mathcal{T}$ for each Borel subset $B$ of $\mathcal{X}$, which is determined by the one-step transition probability $P(\cdot|x,a)$ defined in \Cref{sec:ProblemForm}.

With the definition of the consolidated stage-wise function $\bar{c}$ and the tower property of conditional expectation, the infinite-horizon cost functional in \cref{Eq:CostFunctionalPenalty} can be re-written as 
\begin{equation}\label{Eq:CostFunPenaltyProblem}
f^{\mu}(x) = \mathbb{E}\left[ \sum_{l=0}^\infty \beta^{t_l}\left( \bar{c}(x_{t_l},\mu(x_{t_l})) + \beta^{\tau(x_{t_l})} O )\right)  \middle\vert x_0=x \right].
\end{equation}Define the optimal cost for \Cref{ProblemOfPenalty} as
\begin{equation}\label{Eq:ValueFunPenaltyProblem}
    V_{st} (x)\coloneqq \inf_{\mu\in \Phi_{st}} f^\mu(x),
\end{equation}
where $\Phi_{st}$ is the set of all policies taking the structure of . In the following theorem, we state the DP equation for $V_{st}(\cdot)$.
\begin{theorem}\label{Theo:DPEquation}
The value function defined by \cref{Eq:ValueFunPenaltyProblem} satisfies the following dynamic programming equation:
\begin{equation}\label{Eq:DPEOptimizedLookahead}
\begin{aligned}
V_{st}(x)  = \inf_{a\in\mathcal{A},\Delta t  \in \mathcal{T}} \mathbb{E}\Bigg[&\sum_{t=0}^{\Delta t -1} \beta^t c(x_t,a) +  \beta^{\Delta t} \left(V_{st}(x_{\Delta t}) + O\right)\\
&\Bigg \vert x_0 =x, a_t= a,\forall t<\Delta t\Bigg],
\end{aligned}
\end{equation}
for all $x\in \mathcal{X}$.
If there exists a policy $\mu^* = (\tau^*,\pi^*)$ such that 
$$
\begin{aligned}
V_{st}(x)  = \mathbb{E}\Bigg[&\sum_{t=0}^{\tau^*(x) -1} \beta^t c(x_t,\pi^*(x)) +  \beta^{\tau^*(x)} \left(V_{st}(x_{\tau^*(x)}) + O\right)\\ 
&\Bigg \vert x_0 =x, a_t= \pi^*(x),\forall t<\Delta t\Bigg],\\
\end{aligned}
$$
for all $x\in \mathcal{X}$, then $\mu^*$ is an optimal policy for \Cref{ProblemOfPenalty}. 
\end{theorem}
\begin{proof}
See Appendix \ref{Proof:DPEquation}.
\end{proof}

\begin{remark}
The DP equation in \cref{Eq:DPEOptimizedLookahead} includes the consolidated state costs, $\sum_{t=0}^{\Delta t-1} \beta^t c(x_t,a)$, which is the accumulated costs incurred from the current update time instance to the next update time instance, the cost-to-go after $\Delta t$-steps of lookahead, $\beta^{\Delta t} V(x_{\Delta x})$, and the penalty for a new update $\beta^{\Delta t} O$. Based on the current measurement $x$, the DP equation has $\Delta t$-steps of lookahead. The number of steps $\Delta t$ is optimized in order to balance the trade-off between the system performance and the update penalty. Thus, we refer to \cref{Eq:DPEOptimizedLookahead} as the DP equation with optimized lookahead. The optimized number of lookahead steps is the optimal waiting time for the next triggering given the penalty of triggering $O$. When $O=0$, the DP equations gives $V_{st}(x) = V(x),\forall x\in\mathcal{X}$, i.e., the value function is the same as the one in classic MDPs.
\end{remark}

\begin{remark}[Computational Methods]
One can resort to methods such as the usual value iteration or the policy iteration \cite{puterman2014markov} to solve the DP equation. In the value iteration approach, given the $k$-th estimate of the value function, $V_{st,k}(\cdot)$, the next estimate $V_{st,k+1}$ can be computed using \cref{Eq:DPEOptimizedLookahead}. Repeat this process until it converges to the fixed-point of \cref{Eq:DPEOptimizedLookahead}. The convergence is guaranteed for any given $V_{st,0}$, when $\beta <1$, in view of the Banach fixed-point theorem (see Theorem 6.2.3. of \cite{puterman2014markov}). And the convergence rate is guaranteed to be
$\Vert V_{st,k} - V_{st} \Vert \leq ({\beta^k}/({1-\beta}))\Vert V_{st,0} - V_{st,1} \Vert$. The actual convergence speed should be faster than the above rate depending on what the update penalty $O$ is. 

%When the curse of dimensionality occurs, one can turn to approximation alternatives of the above mentioned methods such as approximate value iteration and approximate policy iteration. When the transition probability is unknown, one can avail oneself of the existing reinforcement learning methods \cite{bertsekas1996neuro}. Tackling the curse of dimensionality and the issue of unknown environment is beyond the scope of this paper but worth further exploring in future endeavors.
\end{remark}

With \Cref{Theo:DPEquation}, we can compute the value function $V_{st}(\cdot)$ and the optimal self-triggering policy $\mu^*$. The computation of $V_{st}(\cdot)$ and $\mu^*$ is usually off-line, and then $\mu^*$ is deployed for online implementation. In the next sub-section, we propose a greedy policy that solves \Cref{ProblemofHardConstraint}, i.e., a policy that reduces the number of updates while maintaining a certain level of sub-optimality.

\subsection{Performance Guaranteed Self-Triggering Policies}\label{Subsec:PerfGuaranteed}
In this sub-section, we propose a greedy self-triggering policy $\mu$ that achieves the inequality defined in \cref{Eq:Sub-optimalityLevel}. To present the policy, we begin with the following lemma.
\begin{lemma}\label{Lemma:NecessaryHardConstraint}
If a self-triggering policy $\mu=(\pi,\tau)$ achieves the following inequality
\begin{equation}\label{Eq:NecessaryHardConstraint}
\mathbb{E}\left[ \sum_{t=0}^{\tau(x)-1} \beta^t c(x_t,\pi(x)) +  \alpha\beta^{\tau(x)} V(x_{\tau(x)})\middle\vert x_0 =x\right] \leq \alpha V(x),
\end{equation}
for all $x\in\mathcal{X}$, then we have $v^\mu(x) \leq \alpha V(x)$.
\end{lemma}
\begin{proof}
See \Cref{Proof:NecessaryHardConstraint}.
\end{proof}

\Cref{Lemma:NecessaryHardConstraint} offers us a convenient way to find an policy that achieves the  performance level specified by $\alpha V(x)$ for all $x\in\mathcal{X}$ and for $\alpha \geq 1$. Since the agent aims to reduce the amount of sensing/communication resources (the rate of updating), he/she needs to find, for each $x\in\mathcal{X}$, the maximum $\Delta t_x\in \mathcal{T}$ such that there exists at least an action $a_x\in\mathcal{A}$ so that \cref{Eq:NecessaryHardConstraint} is satisfied with $\tau(x)$ and $\pi(x)$ replaced by $\Delta t_x$ and $a_x$ respectively. Then, \Cref{ProblemofHardConstraint} becomes solving the following problem for each $x\in\mathcal{X}$
\begin{equation}\label{Eq:OptimizationHardConstraint}
\begin{aligned}
&\max_{\Delta t_x \in \mathcal{T}, a_x\in \mathcal{A}} \Delta t_x\\
&\ \ \ \ \ s.t.\ \ \ \ \textrm{\cref{Eq:NecessaryHardConstraint}},
\end{aligned}
\end{equation}
where in \cref{Eq:NecessaryHardConstraint},  we replace $\tau(x)$ and $\pi(x)$ with $\Delta t_x$ and $a_x$ respectively.
\begin{theorem}\label{Theo:Well-Defineness}
If there exists an optimal policy $\phi^*$ for the classic MDP such that $v^{\phi^*}=V(x)$, then for any fixed $\alpha \geq 1$, there always exist a feasible set for (\ref{Eq:OptimizationHardConstraint}), i.e., the problem (\ref{Eq:OptimizationHardConstraint}) is well-defined.  
\end{theorem}
\begin{proof}
See \Cref{Proof:Well-Defineness}.
\end{proof}
\begin{remark}[The Greedy Choice Property] Note that the self-triggering policy for \Cref{ProblemofHardConstraint} follows the greedy rule. At time $t_l$. the next update time $t_{l+1} = t_l + \tau(x_{t_l})$ is maximized while ensuring \cref{Eq:NecessaryHardConstraint} without considering the effect of this choice on the number of future updates after $t_{l+1}$. Different from the greedy policy, the self-triggering policy $\mu^*$ from \Cref{Theo:DPEquation} for solving \Cref{ProblemOfPenalty} follows the dynamic programming rule, i.e., current choices are made taking into account the 
influence of current choices on the future possibilities.
\end{remark}

So far in this section, we have developed \Cref{Theo:DPEquation} and \Cref{Theo:Well-Defineness} to help find the self-triggering policies that can solve \Cref{ProblemOfPenalty} and \Cref{ProblemofHardConstraint}.  The theorems were developed without specifying the state space $\mathcal{X}$, the action space $\mathcal{A}$, and the transition probabilities, except that we require $\mathcal{X}$ to be Polish and $c(\cdot,\cdot)$ to be bounded and non-negative on $\mathcal{X}\times \mathcal{A}$. Hence, The results are applicable to a variety of models such as LQG control \cite{gommans2014self,akashi2018self,huang2020infinite}, inventory control \cite{feinberg2016optimality}, and queueing systems \cite{wikecek2016mean,huang2019continuous}. The two theorems pave the way for the computation and implementation of the self-triggering policies for various Markov decision processes. In the next section, we present a gridworld example to illustrate the computation and implementation of self-triggering policies using \Cref{Theo:DPEquation} and \Cref{Theo:Well-Defineness}.

\section{Computation and Implementation: A Gridworld Case Study}\label{sec:NumericalStudy}
In this section, we consider a rectangular gridworld representation of a simple MDP for illustration purposes. The gridworld environment made up of $4 \times 6$ cells is shown in \ref{fig:GridWorldMap}, where grey areas are walls. An agent lives in this gridworld aiming to navigate from the start cell to the target cell. The states, representing the cell the agent lives in, are $\mathcal{X}=\{1,2,\cdots,19,20\}$. There are four actions possible at each state, $\mathcal{A}=\{\textit{north},\textit{south},\textit{east},\textit{west}\}$. Walls block the agent's path. The actions that would take the agent off the grid or into the walls in fact leave the state unchanged. State $x=20$ is an absorbing state such that once the agent reaches the target cell, he/she enters the absorbing state with probability one (w.p.1). The agent aims to reach the target as fast as soon. Hence, we define 
\begin{equation}\label{Eq:GridworldCost}
c(x,a) = \begin{cases}
10,\ \ \ \textrm{if }x\in\mathcal{X}\backslash \{19,20\},\\
0,\ \ \ \textrm{if }x\in\{19,20\}.
\end{cases}
\end{equation}

\subsection{A Non-Windy Gridworld}
We first consider a non-windy setting where each action deterministically causes the agent to move one cell in the respective direction. Let $ P^d$ denotes the transition probabilities in a non-wind setting. For instance, we have $P^d(6|1,\textit{north}) = 1$. We consider the discount factor $\beta=0.95$, and the bound on the waiting time for the next update is $\bar{T}=6$. The update penalty $O$ is subject to change.
\begin{figure}
    \centering
    \includegraphics[width=0.8\columnwidth]{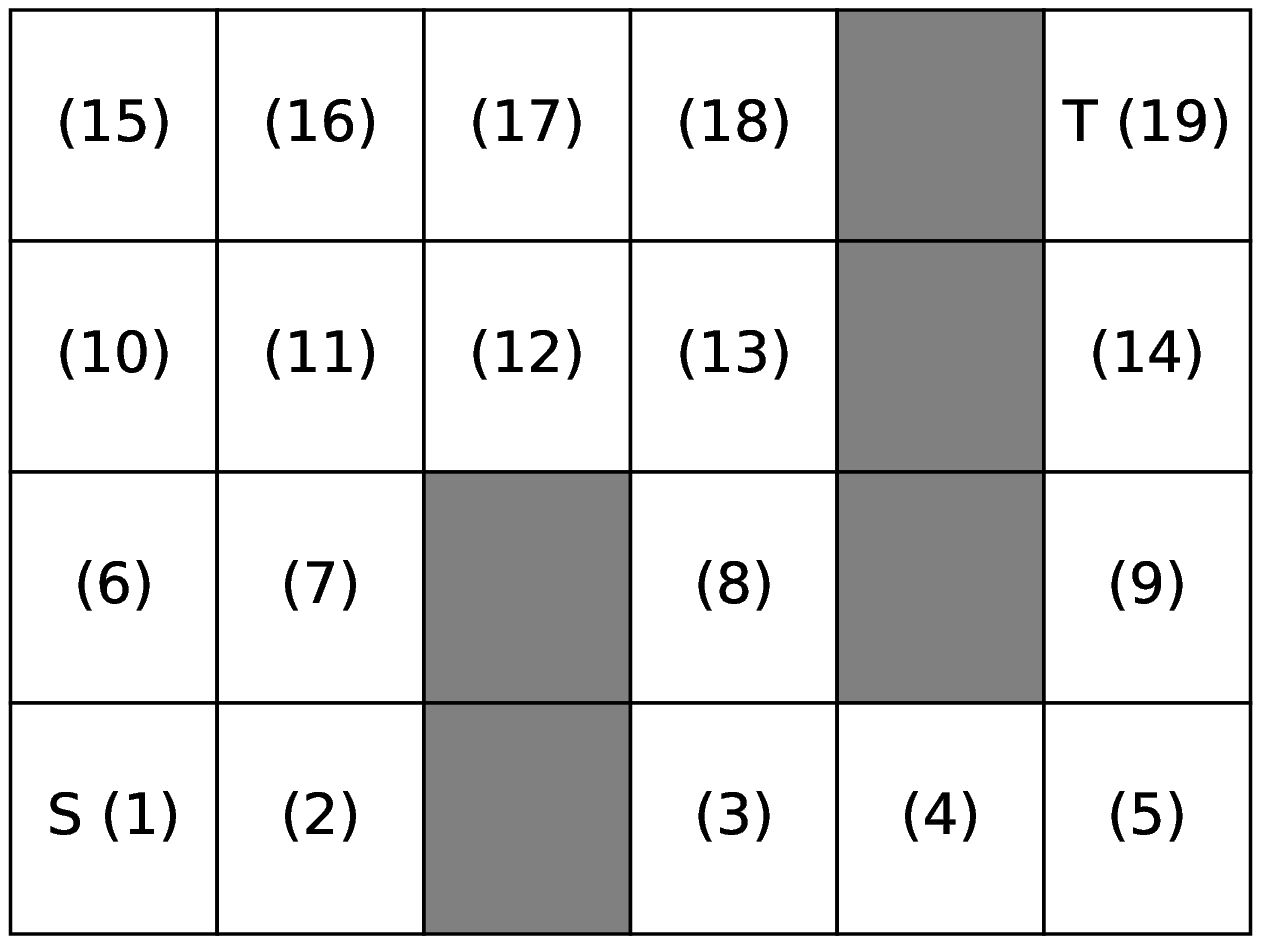}
    \caption{A gridworld example: Grey areas represent walls, \textit{S} stands for the start cell, \textit{T} denotes the target cell, and the integers in the brackets are the indices of states.}
    \label{fig:GridWorldMap}
\end{figure}

We set the initial value function estimate to be $V_{st,0}(x) =0, \forall x\in\mathcal{X}$. We conduct value iteration using the DP equation with controlled lookahead in \cref{Eq:DPEOptimizedLookahead}: 
\begin{equation*}
\begin{aligned}
V_{st,k+1}(x)  = \min_{a\in\mathcal{A},\Delta t  \in \mathcal{T}} \mathbb{E}\Bigg[&\sum_{t=0}^{\Delta t -1} \beta^t c(x_t,a) +  \beta^{\Delta t} \big(V_{st,k}(x_{\Delta t}) \\
&+O\big) \Bigg \vert x_0 =x, a_t= a,\forall t<\Delta t\Bigg],
\end{aligned}
\end{equation*}
where every term in the expectation operator can be computed using transition probabilities $P^d$. The iteration stops when $\Vert V_{st,k+1} -V_{st,k} \Vert\leq 10^{-5}$, and the results show that the tolerance can be achieved within $25$ iterations for every update penalties $O$ we study in this paper. 

\begin{figure}
    \centering
    \includegraphics[width=0.8\columnwidth]{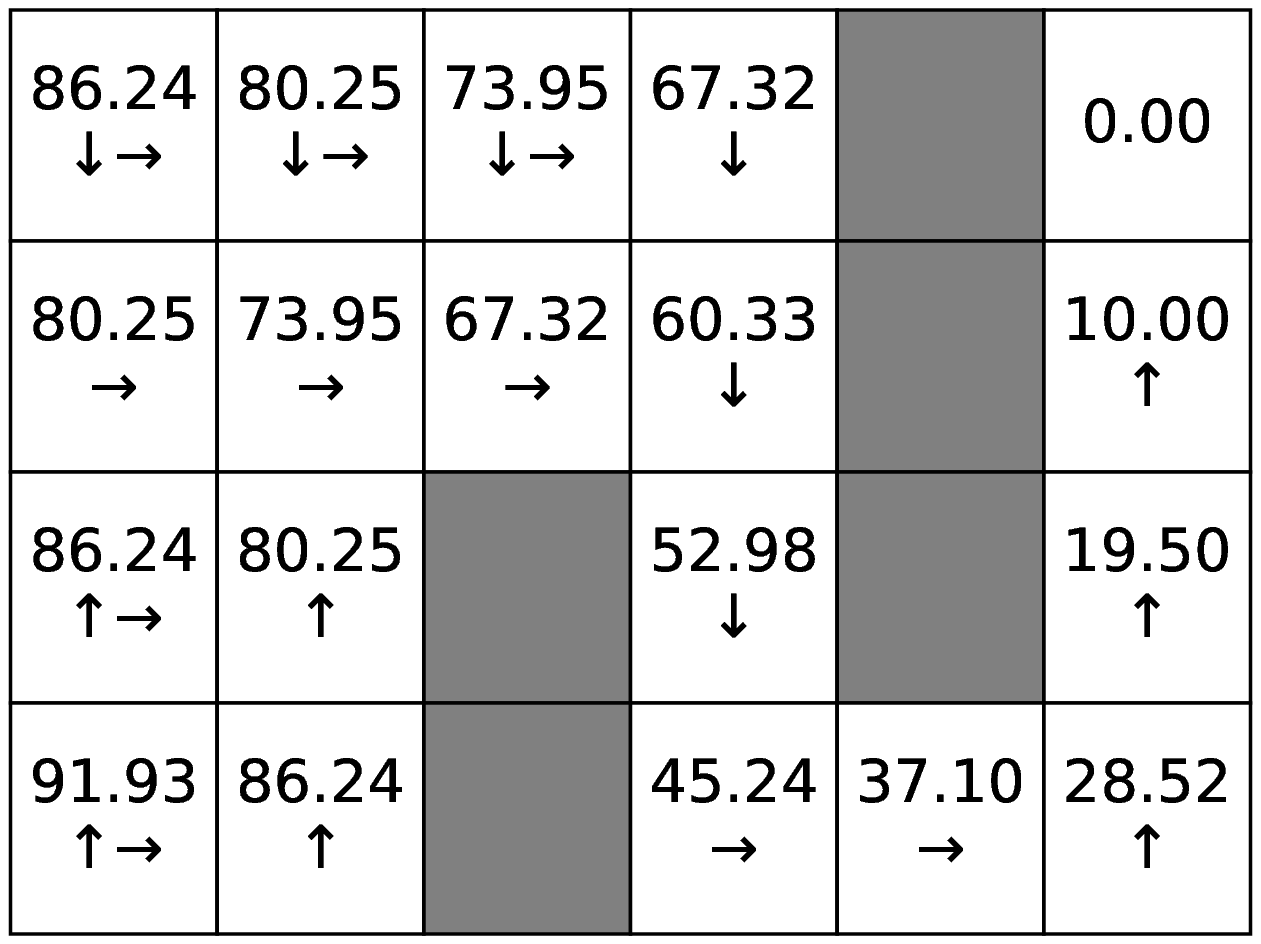}
    \caption{A non-windy gridworld: The value $V(x)$ (the upper value) and the optimal action $\phi^*(x)$ (the lower pointers) in the classic MDP for each $x\in X$.}
    \label{fig:ClassicMDPValueAction}
\end{figure}

\begin{figure}
    \centering
    \begin{subfigure}{0.49\columnwidth}
        \centering
        \includegraphics[width=\columnwidth]{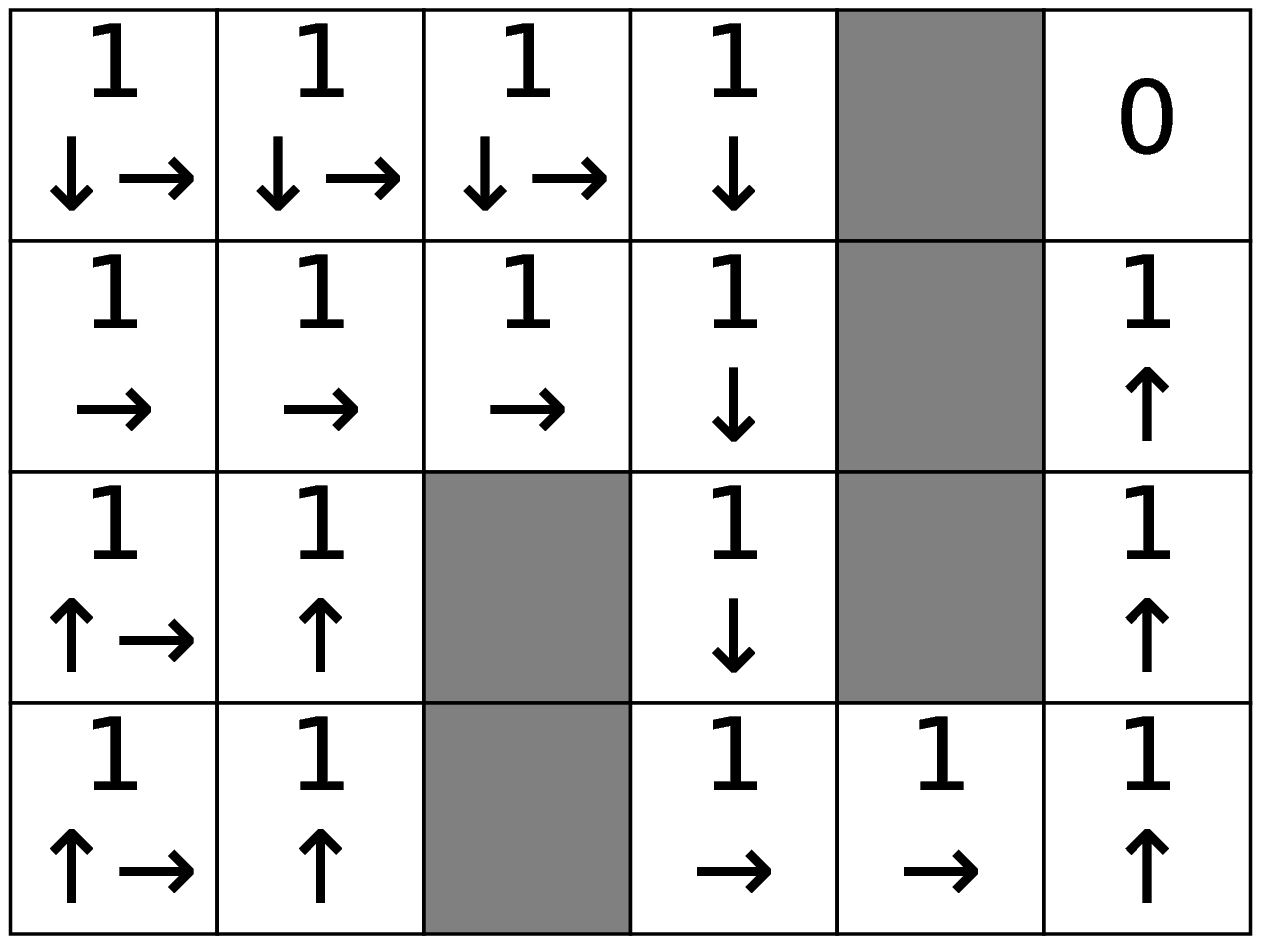}
        \caption{{\small The update penalty $O=0$.}}    
        %\label{fig:non-windyActions-a}
    \end{subfigure}
    \hfill
    \begin{subfigure}{0.49\columnwidth}  
        \centering 
        \includegraphics[width=\columnwidth]{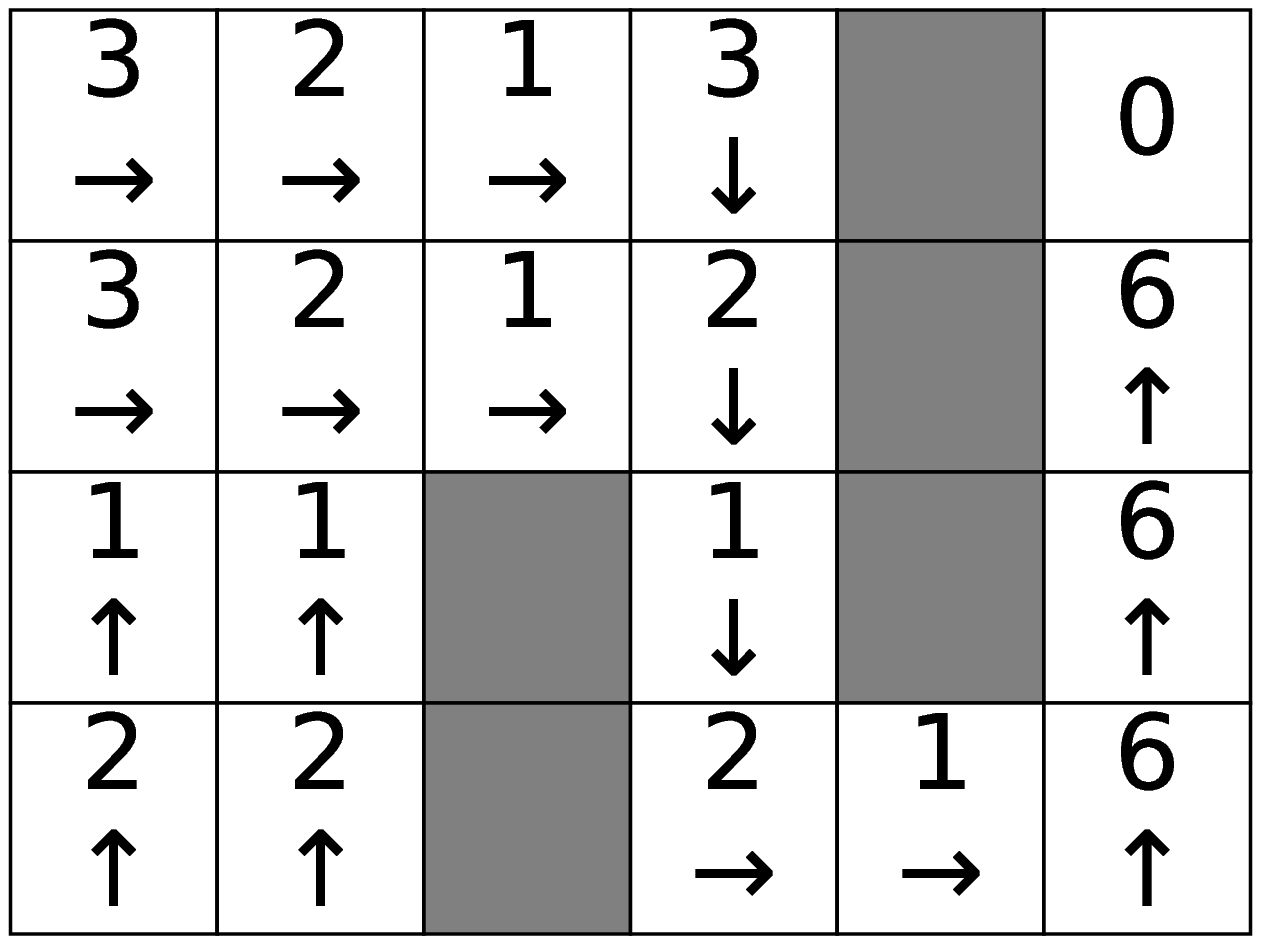}
        \caption{{\small  The update penalty $O=0.1$.}}    
        %\label{fig:non-windyActions-b}
    \end{subfigure}
    \vskip\baselineskip
    \begin{subfigure}{0.49\columnwidth}   
        \centering 
        \includegraphics[width=\columnwidth]{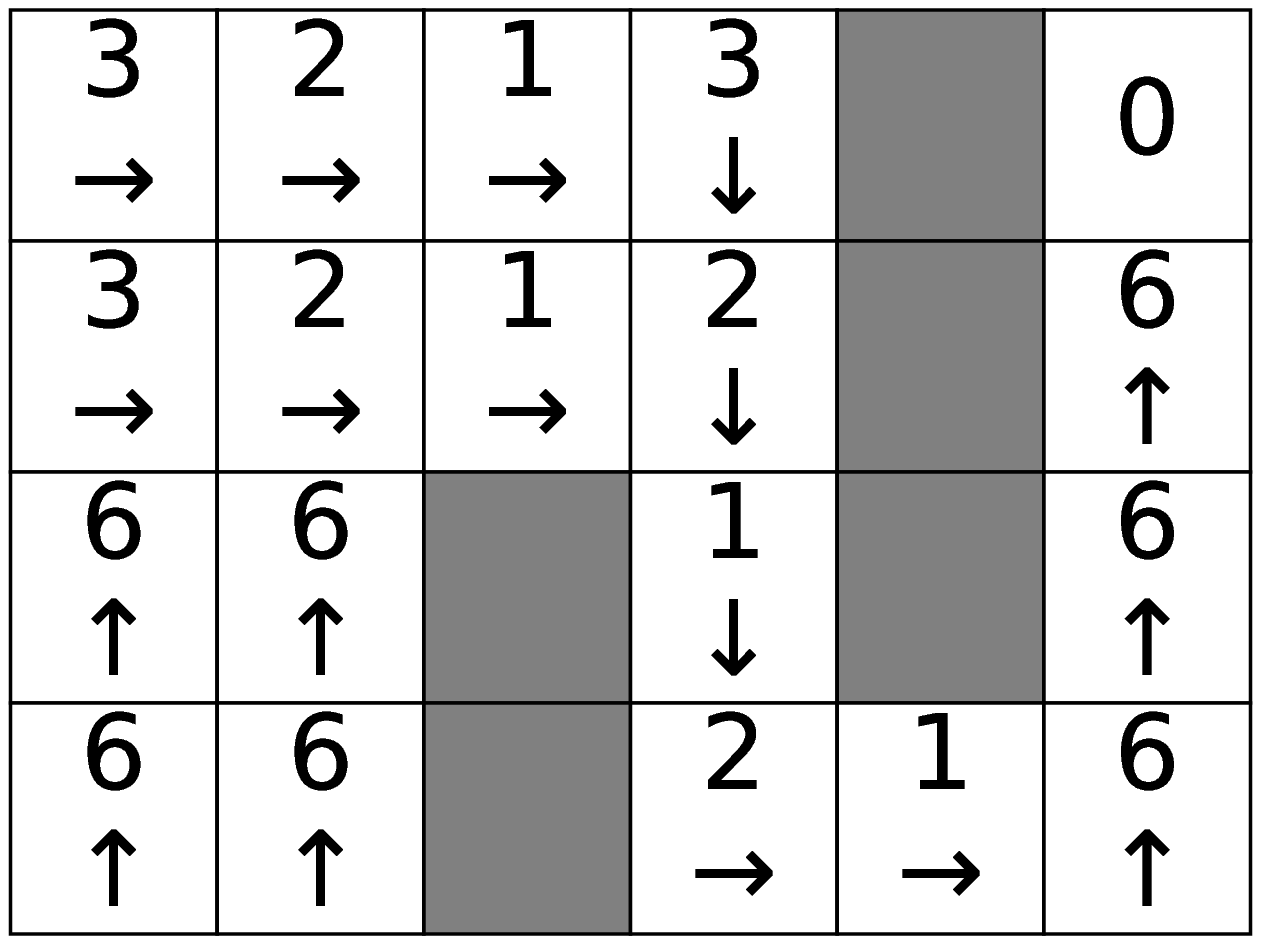}
        \caption{{\small The update penalty $O=40$.}}    
        %\label{fig:non-windyActions-c}
    \end{subfigure}
    \hfill
    \begin{subfigure}{0.49\columnwidth}   
        \centering 
        \includegraphics[width=\columnwidth]{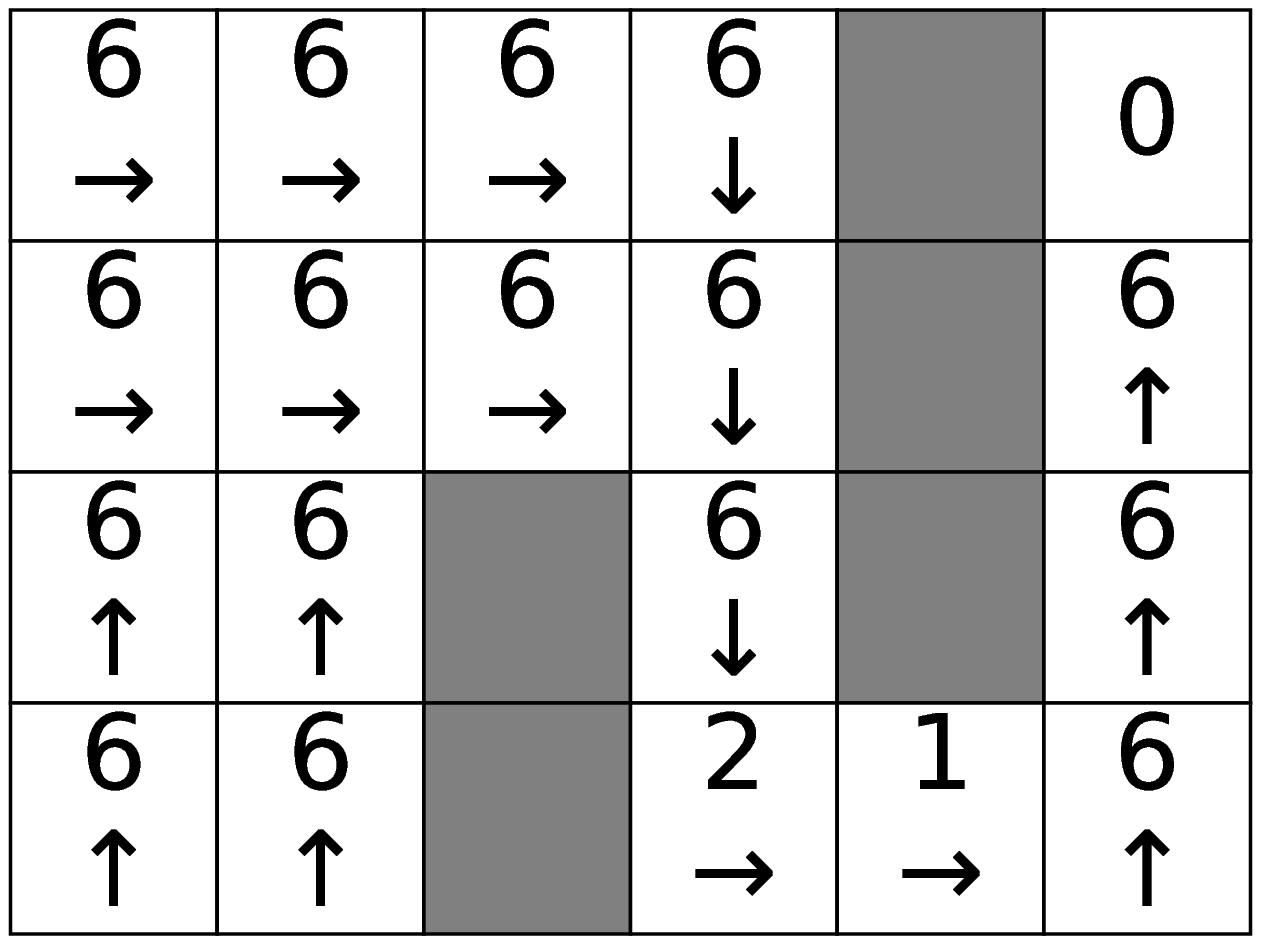}
        \caption{{\small The update penalty $O=80$.}}    
        %\label{fig:non-windyActions-d}
    \end{subfigure}
    \caption{A non-windy gridworld: The optimal triggering time policy $\tau^*(x)$ (the upper value) and the optimal control policy $\pi^*(x)$ (the lower pointers) for each $x\in\mathcal{X}$ under different update penalties $O$. (For \Cref{ProblemOfPenalty})} 
    \label{fig:non-windyActions}
\end{figure}

In Fig. \ref{fig:non-windyActions}, we present the optimal triggering time $\tau(x)^*$ and the optimal control policy $\pi^*(x)$ for each state when the update penalties are $O=0,0.1,40,80$. As we can see from Fig. \ref{fig:non-windyActions} (a), when $O =0$, since there is no update penalty, the optimal triggering time is to update every time, i.e., $\tau^*(x)=1,\forall x\in\mathcal{X}$, and the optimal control policy is the same as its counterpart in a classic setting, i.e., $\pi^*(x) = \phi^*(x),\forall x\in \mathcal{X}$. The policy offers three paths from the start cell to the target cell: $1\rightarrow 2 \rightarrow 7 \rightarrow 11 \rightarrow \cdots\rightarrow 19$, $1 \rightarrow 6 \rightarrow 7 \rightarrow 11 \rightarrow \cdots \rightarrow 19$, and  $1\rightarrow 6 \rightarrow 10 \rightarrow 11 \rightarrow \cdots \rightarrow 19$. Each path takes $12$ steps to complete, covers $13$ cells, and there are $12$ updates. 

Suppose a remote controller controls the agent, and the communication between them is expensive. Each communication/update induces an update penalty $O$. When $O= 0.1$, as is shown in \ref{fig:non-windyActions} (b), an update is only triggered when there is a need to update the action. For example, when the agent is at state $x=1$ at time $0$, the optimal control policy is heading north, and the optimal waiting time is $2$ steps. That means at time $t=0$, the agent communicates with the controller and is commanded to go north and fix this action for $2$ time steps, after which a new update will be sent. Since there is a straight path to the target cell, in a non-windy setting, at states $x=8,9,14$, the controller chooses the maximum allowed waiting time $\bar{T}=6$. There are few points worth noticing when we compare Fig. \ref{fig:non-windyActions} (a) and (b): First, when the update penalty $O=0.1$, the optimal policy, as is shown in Fig.  \ref{fig:non-windyActions}, provides one shortest path to the target cell: $1 \rightarrow 6 \rightarrow 10 \rightarrow 11 \rightarrow \cdots \rightarrow 19$. The path takes $12$ time steps to complete, which is the same as when $O=1$. However, the updates are only triggered when the agent was at states $x = 10,16,3,5$. Hence, the self-triggering policy under $O=0.1$ requires only $4$ updates to achieve the same shortest path as the classic optimal policy. That means the self-triggering policy saves $(12-4)/12 = 66.47\%$ of the communication resources required in a classic policy. Second, When the update penalty is $O=0.1$, at state $x=1$, going west is no longer an optimal choice since going west requires more updates (5 in this case) to achieve the shortest path.
Third, in Fig. \ref{fig:non-windyActions} (b), the optimal triggering time and the corresponding optimal control at each state always take the agent to the next turning points. For instance, at $x=10$, the optimal action is to go east and to fix this direction for $3$ steps. This optimal action and optimal waiting time take the agent to state $13$, where the agent has to turn south to reach the target cell. There are two reasons to explain this phenomenon: 1. the update penalty is relatively low, compared with the stage cost defined in \cref{Eq:GridworldCost}, so that achieving the shortest path within the minimum number of steps is still a priority. 2. In a non-windy setting, the actions deterministically move the agent toward the desired direction, which means the controller can anticipate the agent's trajectory in future steps. Hence, no update is needed between the two turning points.

The computed self-triggering policy under $O=40$ is provided in Fig. \ref{fig:non-windyActions} (c). The self-triggering policy gives a longer path to reduce the overhead of updating: $1 \rightarrow \cdots \rightarrow 15 \rightarrow \cdots \rightarrow 18 \rightarrow 3 \rightarrow \cdots \rightarrow 19$,  which takes $17$ time steps (stay at state $15$ for $4$ time steps due to $6$ time steps of going north without update), covers $15$ cells, and requires $4$ updates to complete. Even though the self-triggering policy requires the same number of updates as the case when $O=0.1$, the updates are triggered later than their counterparts in the case of $O=0.1$. Hence, the updates produce less costs due to the discount effect. As the update penalty increases to $O=80$ (see Fig. \ref{fig:non-windyActions} (d)), the optimal time policy at most of the states becomes to wait as long as possible for next update, i.e., $\tau^*(x) = \bar{T}$, for $x\in\mathcal{X}\backslash \{3,4\}$.

\subsection{A Windy Gridworld}

Next, we consider a windy gridworld where the wind takes the agent north $10\%$ of the chance and west $10\%$ of the chance. And $80\%$ of the time, the agent's movement follows its action. In the windy gridworld, the effect of boundaries and walls still applies. The transition probability in a windy setting is defined by $P^w$. For example, if the agent is at state $x=11$ and chooses to go east, we have $P^w(12|11,\textit{east})= 0.8$, $P^w(10|11,\textit{east})= 0.1$, and $P^w(16|11,\textit{east})= 0.1$. We run value iterations using the DP equation with controlled lookahead given in \cref{Eq:DPEOptimizedLookahead} under the transition probabilities $P^w$ in the windy environment.

The optimal timing policy and optimal control policy are presented in Fig. \ref{fig:windyActions}. One difference in a windy environment is that the control chosen will not deterministically cause the movement of the agent. That means if there is no update, the controller needs to estimate the agent's trajectory, and there exists an estimation error. Hence, we hypothesize that the agent needs to trigger the update more frequently than in a non-windy environment to know his/her location and then adjust his/her control.

Fig. \ref{fig:windyActions} (a) presents the case when there is no update penalty, i.e., $O=0$. The optimal timing policy is to observe/update every step. The control at state $x=6$ becomes going east to avoid being taken to the northwest corner by the wind. At states $x= 15,16,17$, going south is not an optimal control anymore since if the agent goes south, there is a chance that the wind would take the agent back to the north. When the update penalty is small, i.e., $O=0.1$, the optimal policy is listed in Fig. \ref{fig:windyActions} (b). There are two points worth mentioning when we compare the windy setting and the non-windy setting:
\begin{enumerate}
    \item When $O = 0.1$, the agent updates more frequently in a windy setting. For example, at $x=5$, the agent will update the next step in a windy setting, while the agent will update $6$ steps later in a non-windy setting. One of the reasons is that in a windy setting, the agent has to update in the next step to make sure he/she goes to state $x=9$ instead of being blown by the wind to state $x=4$. This result backs up our hypothesis that the agent in a windy world needs to trigger the update more frequently than in a non-windy environment. 
    \item When $O=40$, Fig. \ref{fig:windyActions} (c) shows some interesting and unexpected results.  The agent waits longer for the next update in a windy setting than in a non-windy setting shown in Fig. \ref{fig:non-windyActions} (c). This result contradicts our hypothesis that the agent tends to update more frequently in a noisy environment. For example, if at time $t$, the agent is at state $11$, the next time the agent will update is $t+6$, which is longer than its counterpart in  Fig. \ref{fig:non-windyActions} (c). One explanation is that since the control is to head east, and the wind pushes the agent north or west, there is no need for the agent to update its action. Eventually, the agent will be more likely to be at state $18$ or $13$ after $6$ steps of fixing his/her control of going east.
\end{enumerate}
When $O=80$, the optimal time policy at every step increases to the maximum allowed waiting time $\bar{T}=6$ to reduce the update penalties.

\begin{figure}
    \centering
    \begin{subfigure}{0.49\columnwidth}
        \centering
        \includegraphics[width=\columnwidth]{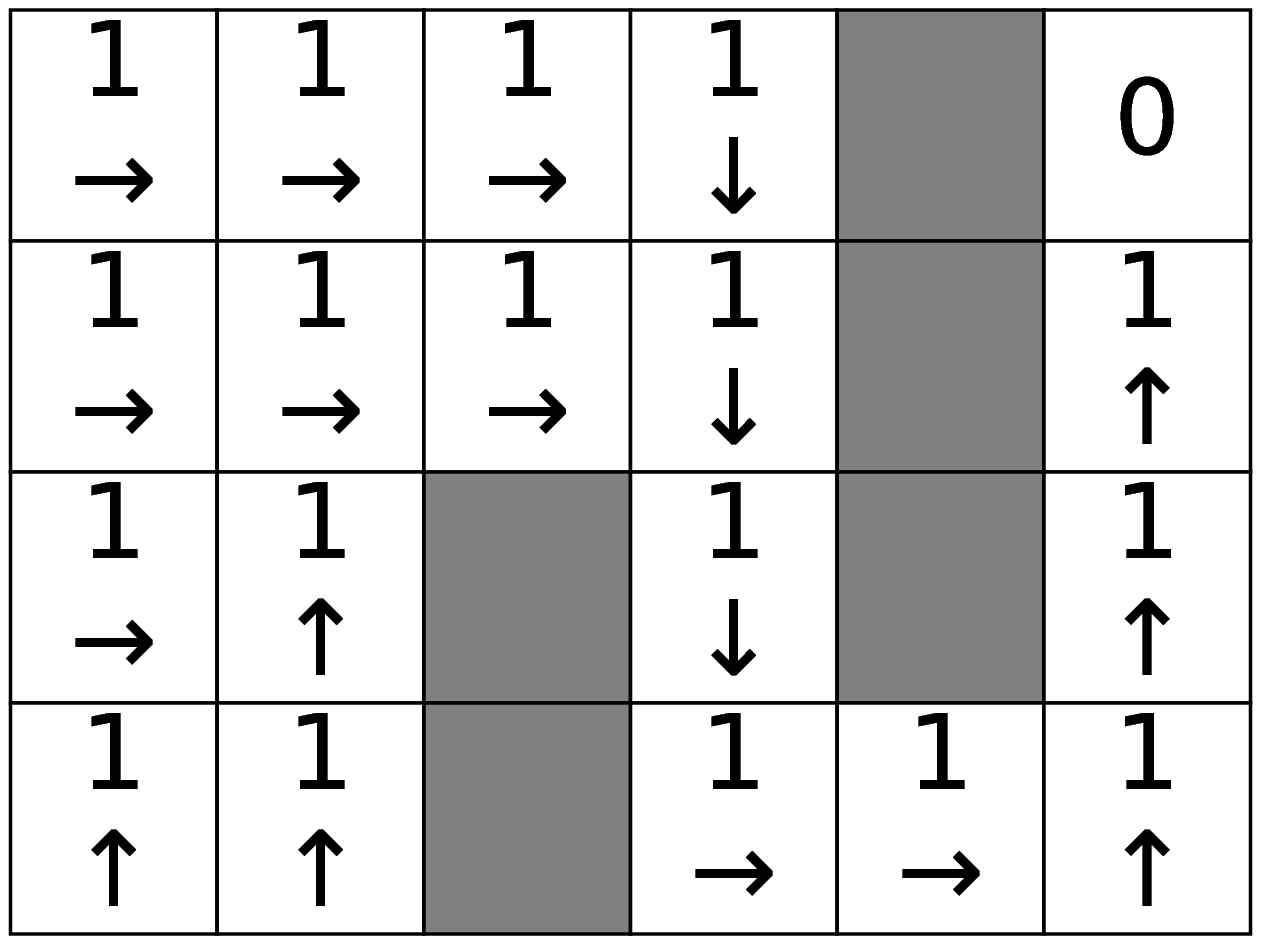}
        \caption{{\small The update penalty $O=0$.}}    
        %\label{fig:non-windyActions-a}
    \end{subfigure}
    \hfill
    \begin{subfigure}{0.49\columnwidth}  
        \centering 
        \includegraphics[width=\columnwidth]{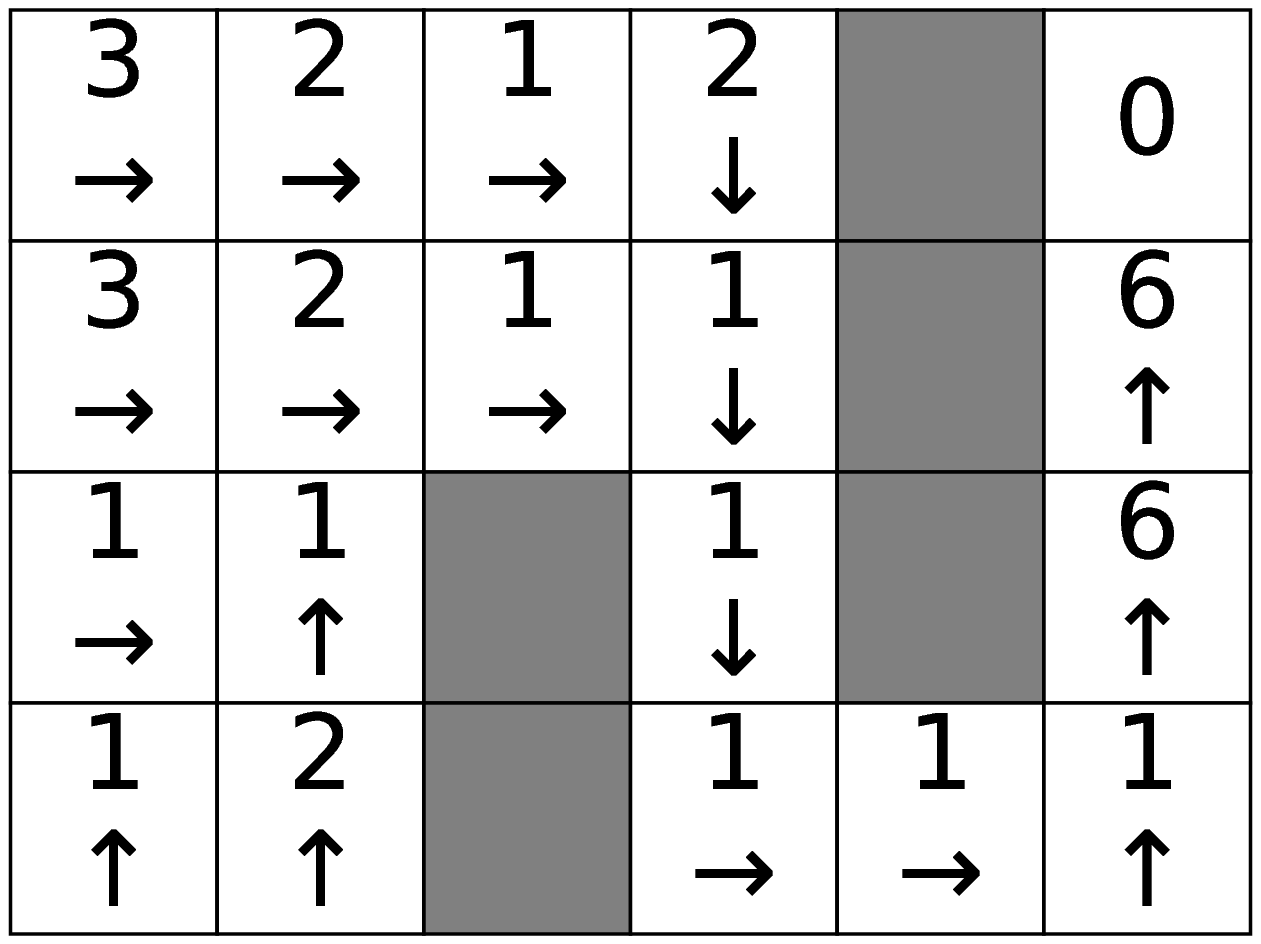}
        \caption{{\small  The update penalty $O=0.1$.}}    
        %\label{fig:non-windyActions-b}
    \end{subfigure}
    \vskip\baselineskip
    \begin{subfigure}{0.49\columnwidth}   
        \centering 
        \includegraphics[width=\columnwidth]{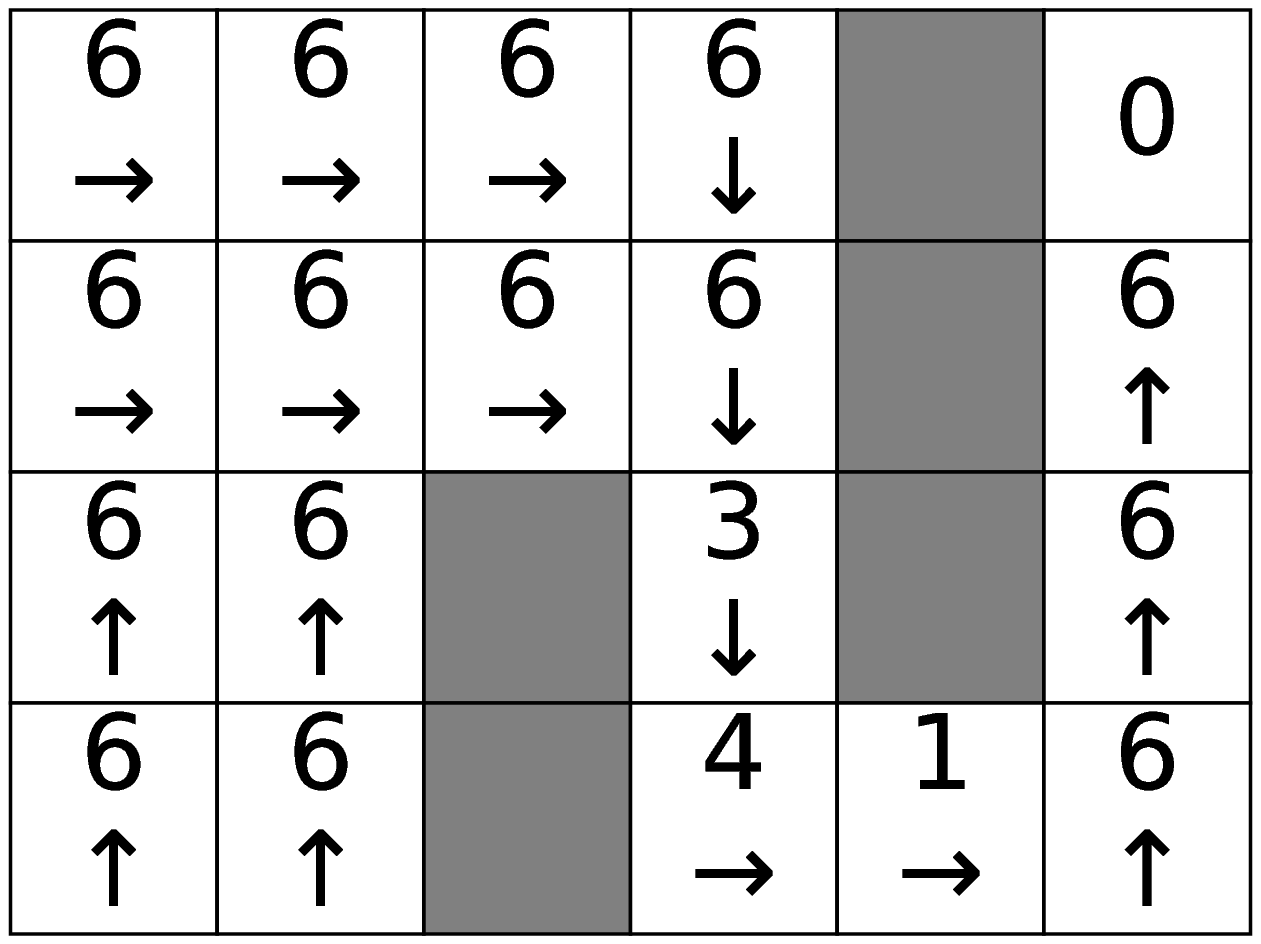}
        \caption{{\small The update penalty $O=40$.}}    
        %\label{fig:non-windyActions-c}
    \end{subfigure}
    \hfill
    \begin{subfigure}{0.49\columnwidth}   
        \centering 
        \includegraphics[width=\columnwidth]{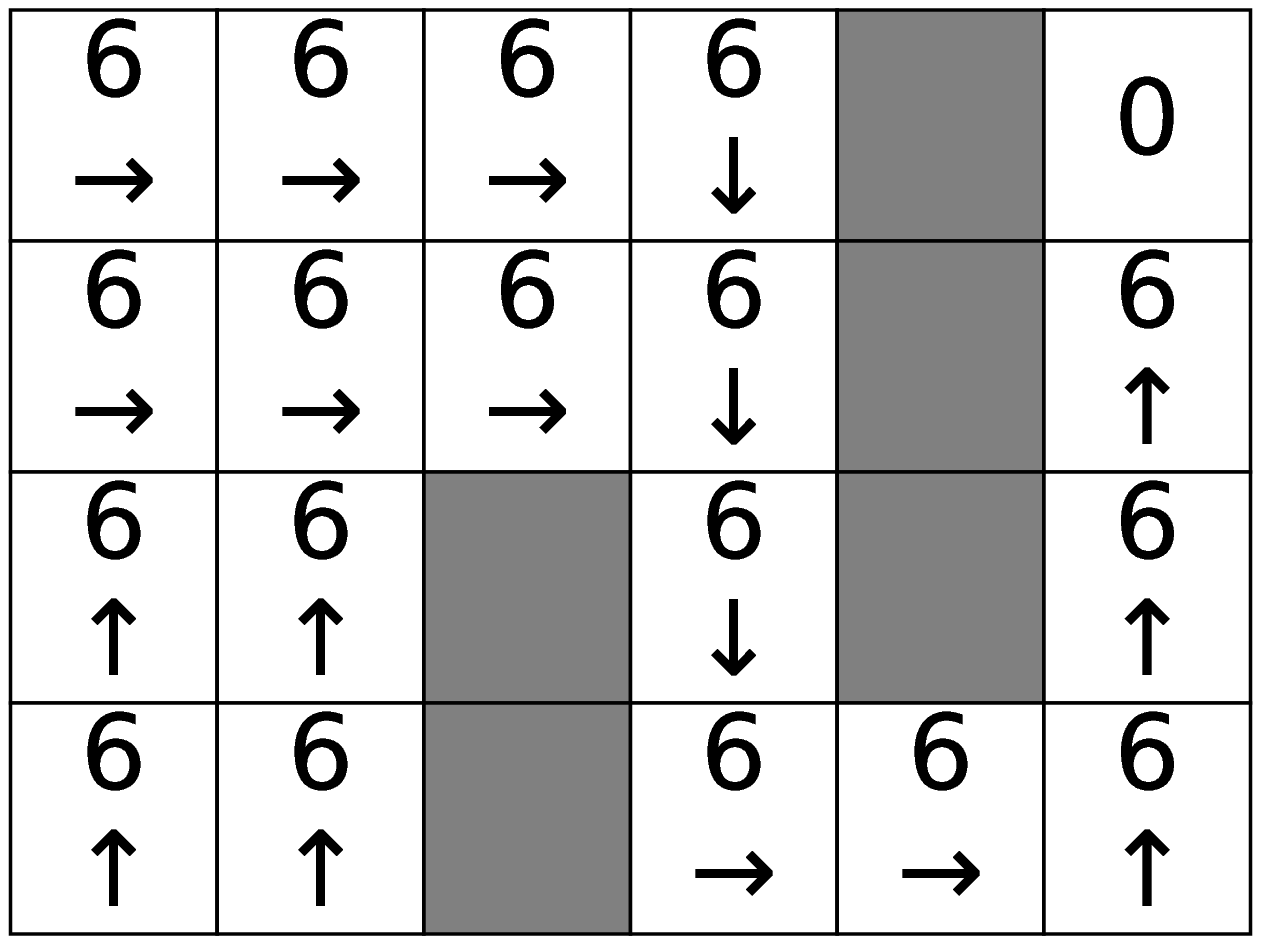}
        \caption{{\small The update penalty $O=80$.}}    
        %\label{fig:non-windyActions-d}
    \end{subfigure}
    \caption{A windy gridworld: The optimal triggering time policy $\tau^*(x)$ (the upper value) and the optimal control policy $\pi^*(x)$ (the lower pointers) for each $x\in\mathcal{X}$ under different update penalties $O$. (For \Cref{ProblemOfPenalty})} 
    \label{fig:windyActions}
\end{figure}

\subsection{Performance Guaranteed Policies}

In the previous subsections, we solve \Cref{ProblemOfPenalty} in the context of a gridworld and obtains the optimal self-triggering policy $\mu^*=(\tau^*,\pi^*)$. Just to remind that we have $\phi:\mathcal{X}\rightarrow \mathcal{A}$, which is the policy in the classic setting, and the self-triggering policy $\mu:\mathcal{X}\rightarrow \mathcal{T}\times \mathcal{A}$ in self-triggered MDPs. To differentiate the self-triggering policy we obtain for \Cref{ProblemOfPenalty} and the policy for \Cref{ProblemofHardConstraint}, we name them $\mu_1^* = (\tau^*_1, \pi_1^*)$ and $\mu_2^* = (\tau_2^*,\pi_2^*)$ respectively.

The self-triggering policy $\mu_1^*$ is optimal with respect to a specified update penalty $O$. However, it does not provide an explicit performance guarantee under the original cost criterion. Instead, the self-triggering policy $\mu_2^*$ provides a pre-specified level of performance guarantee.

As a result of the discussions in \Cref{Subsec:PerfGuaranteed}, the steps to compute a self-triggering policy $\mu_2^*$ for \Cref{ProblemofHardConstraint} is given as follows: 
\begin{enumerate}
\item Compute the value function $\{V(x), x\in\mathcal{X}\}$ of the MDP in the classic setting. 
\item For each $x\in\mathcal{X}$, select $\Delta t_x =\bar{T}$.
\item Compute 
{\small
\begin{equation}\label{ComputationStepsProblem2}
\begin{aligned}
\tilde{V}(x) =
\min_{a\in\mathcal{A}} \mathbb{E}\left[ \sum_{t=0}^{\Delta t_x-1} \beta^t c(x_t,a) +  \alpha\beta^{\Delta t_x} V(x_{\Delta t_x})\middle\vert x_0 =x\right],\\
a^*_x =
\arg\min_{a\in\mathcal{A}} \mathbb{E}\left[ \sum_{t=0}^{\Delta t_x-1} \beta^t c(x_t,a) +  \alpha\beta^{\Delta t_x} V(x_{\Delta t_x})\middle\vert x_0 =x\right].
\end{aligned}
\end{equation}
}
\item If $\tilde{V}(x) > \alpha V(x)$, set $\Delta t_x = \Delta t_x -1$, repeat step 3). Otherwise, $\tau_2^*(x) = \Delta_x$, $\pi^*_2(x) = a_x^*$.
\end{enumerate}

The optimization problem in \cref{ComputationStepsProblem2} admits a closed-form solution for models such as LQG control \cite{gommans2014self} and inventory control \cite{feinberg2016optimality}. For the windy gridworld model, we compute self-triggering policies following the steps for various levels of sub-optimality. The results are presented in Fig. \ref{fig:WindyActions_Problem2}. As we can see from Fig. \ref{fig:WindyActions_Problem2} (a), the self-triggering policy $\mu_2^*$ can achieve a full level of optimality, i.e., $v^{\mu_2^*}(x) =  V(x),\forall x\in\mathcal{X}$, while requiring less communication/sensing resources. When the level of sub-optimality is $\alpha =1.1$, as one can see from Fig. \ref{fig:WindyActions_Problem2} (b), at most states, the optimal timing policy is to wait for two or more than two steps for the next update. That means the self-triggering policy $\mu_2^*$ can save more than $50\%$ communication/sensing resources while suffering only $10\%$ of performance degradation. If one can tolerate a higher level of degradation, one can set $\alpha$ to a higher value and compute the corresponding self-triggering policy $\mu_2^*$. The cases when $\alpha=1.4$ and $\alpha=2$ are presented in Fig. \ref{fig:WindyActions_Problem2} (c) and (d). As one expects, the higher $\alpha$ is (more performance degradation one can tolerate), the fewer updates needed (less communication/resources consumed). 

\begin{figure}
    \centering
    \begin{subfigure}{0.49\columnwidth}
        \centering
        \includegraphics[width=\columnwidth]{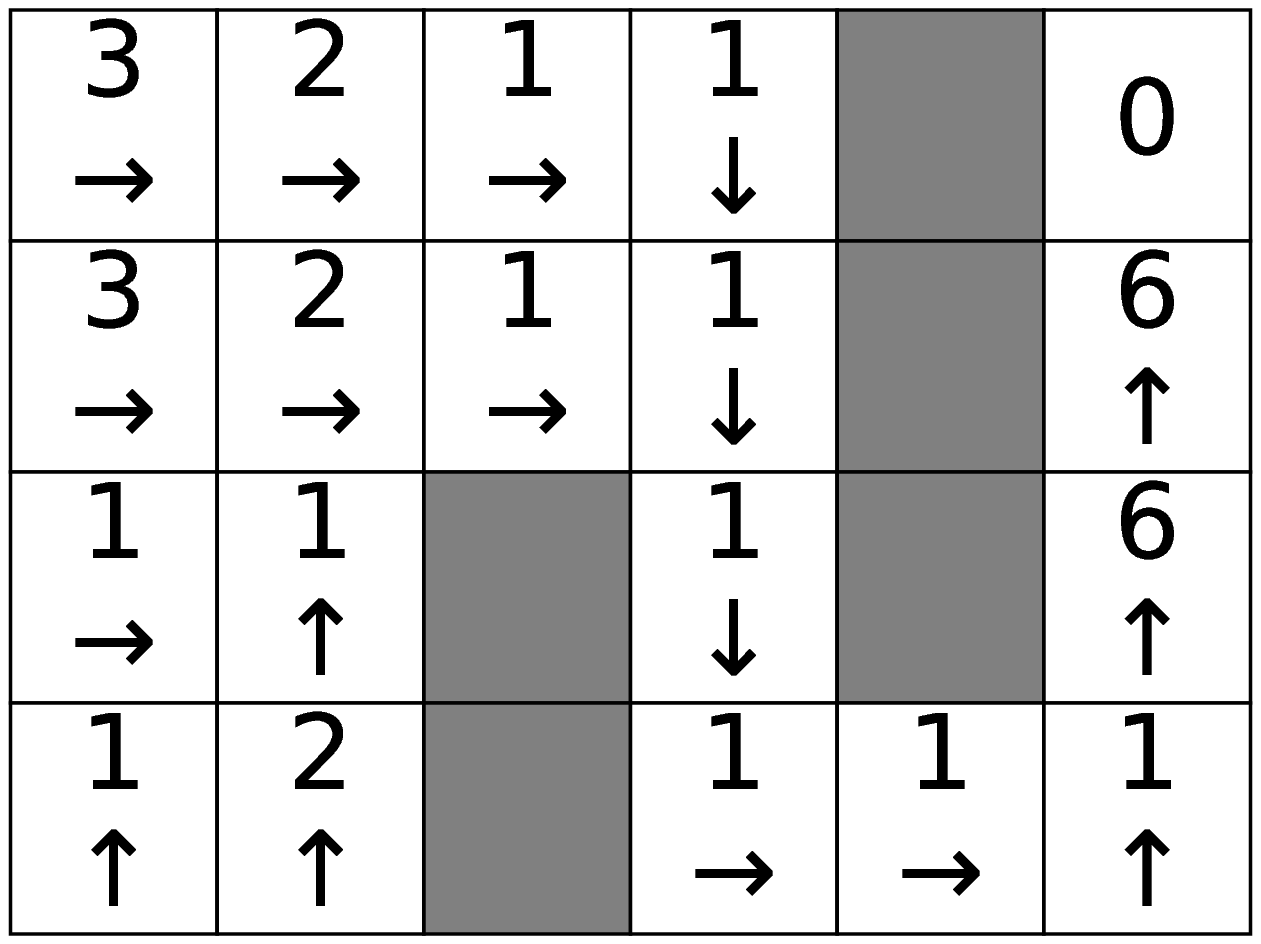}
        \caption{{\small The pre-specified level of sub-optimality penalty $\alpha = 1$.}}    
        %\label{fig:non-windyActions-a}
    \end{subfigure}
    \hfill
    \begin{subfigure}{0.49\columnwidth}  
        \centering 
        \includegraphics[width=\columnwidth]{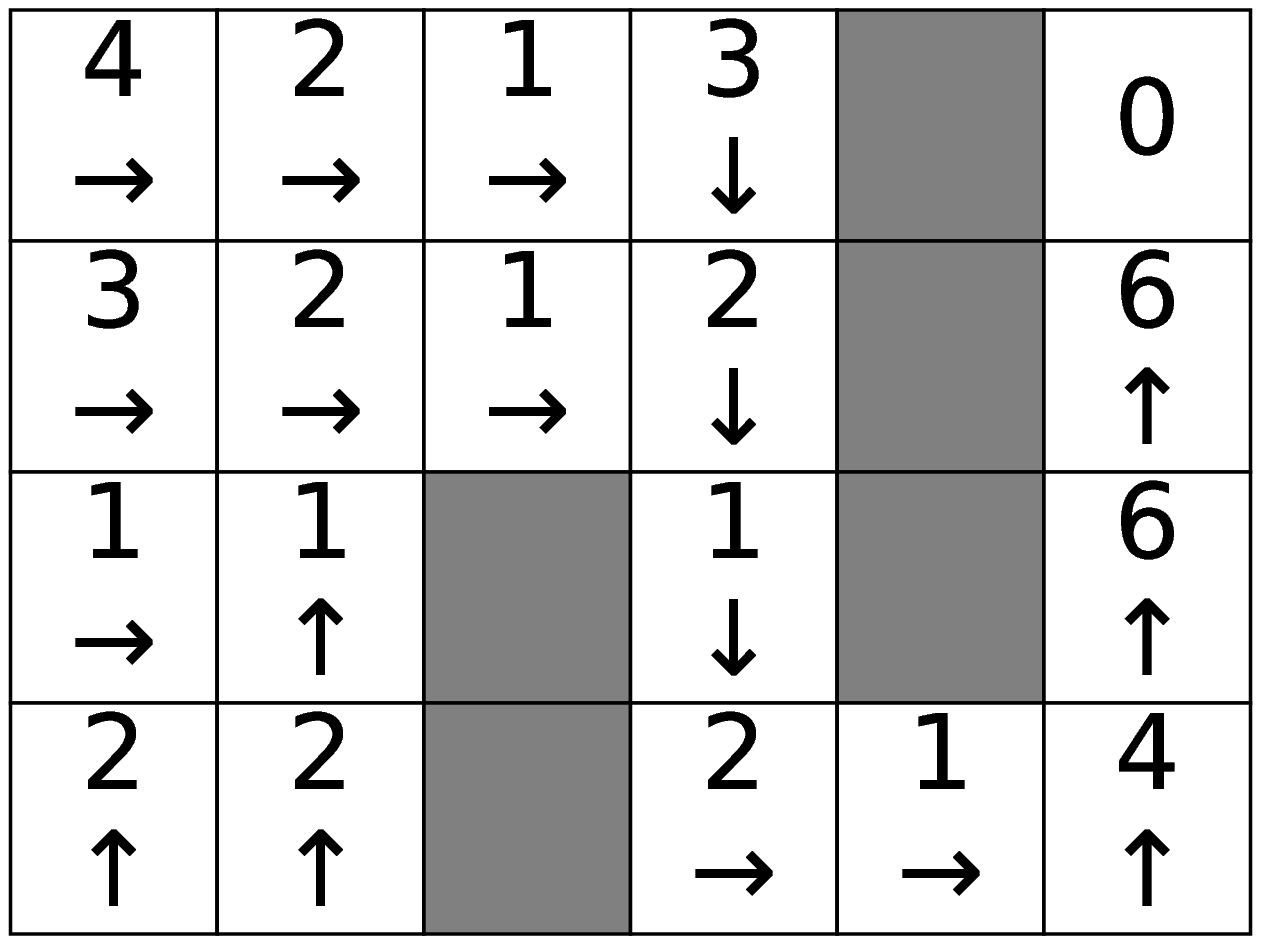}
        \caption{{\small  The pre-specified level of sub-optimality penalty $\alpha = 1.1$.}}    
        %\label{fig:non-windyActions-b}
    \end{subfigure}
    \vskip\baselineskip
    \begin{subfigure}{0.49\columnwidth}   
        \centering 
        \includegraphics[width=\columnwidth]{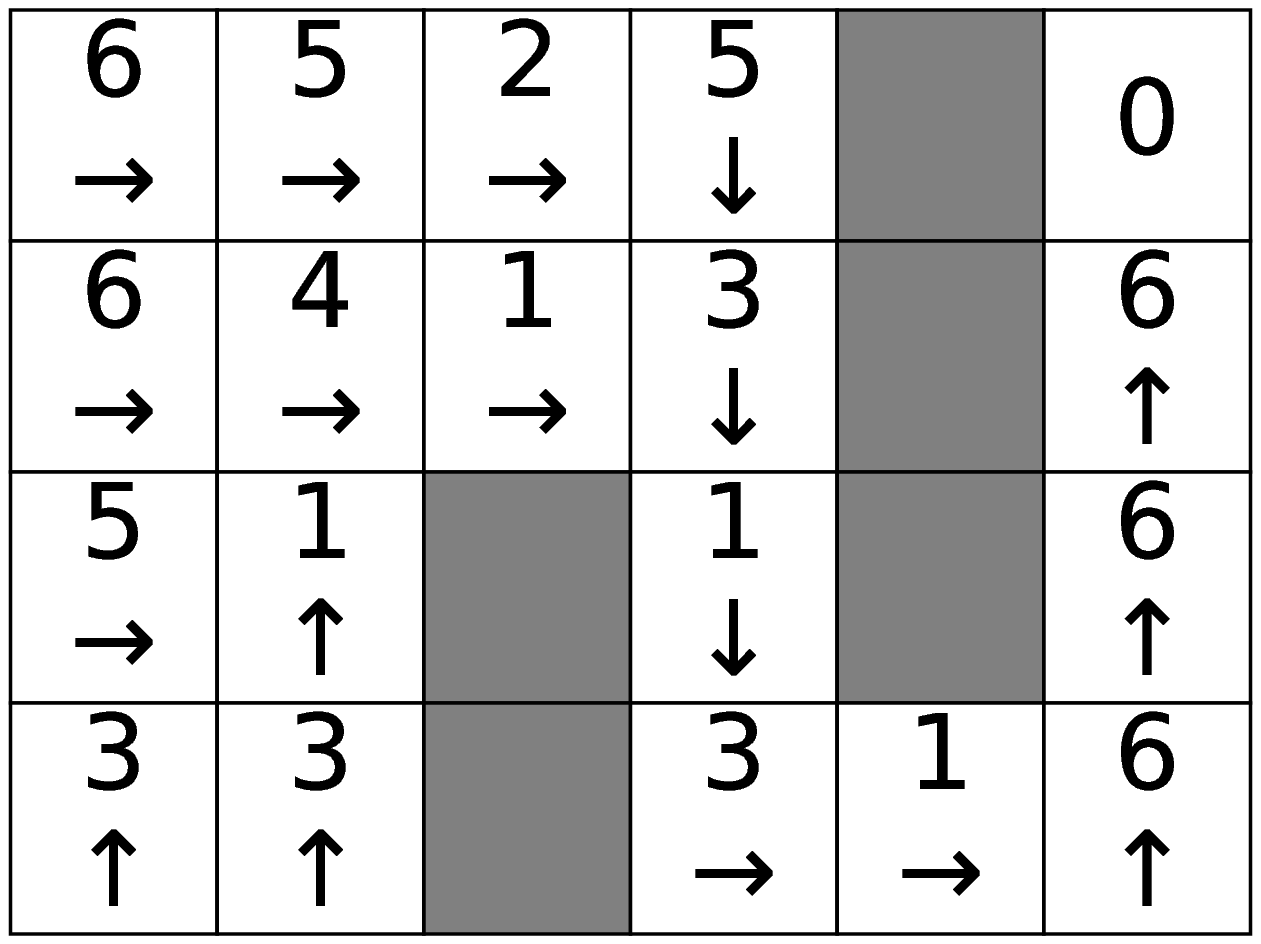}
        \caption{{\small The pre-specified level of sub-optimality penalty $\alpha = 1.4$.}}    
        %\label{fig:non-windyActions-c}
    \end{subfigure}
    \hfill
    \begin{subfigure}{0.49\columnwidth}   
        \centering 
        \includegraphics[width=\columnwidth]{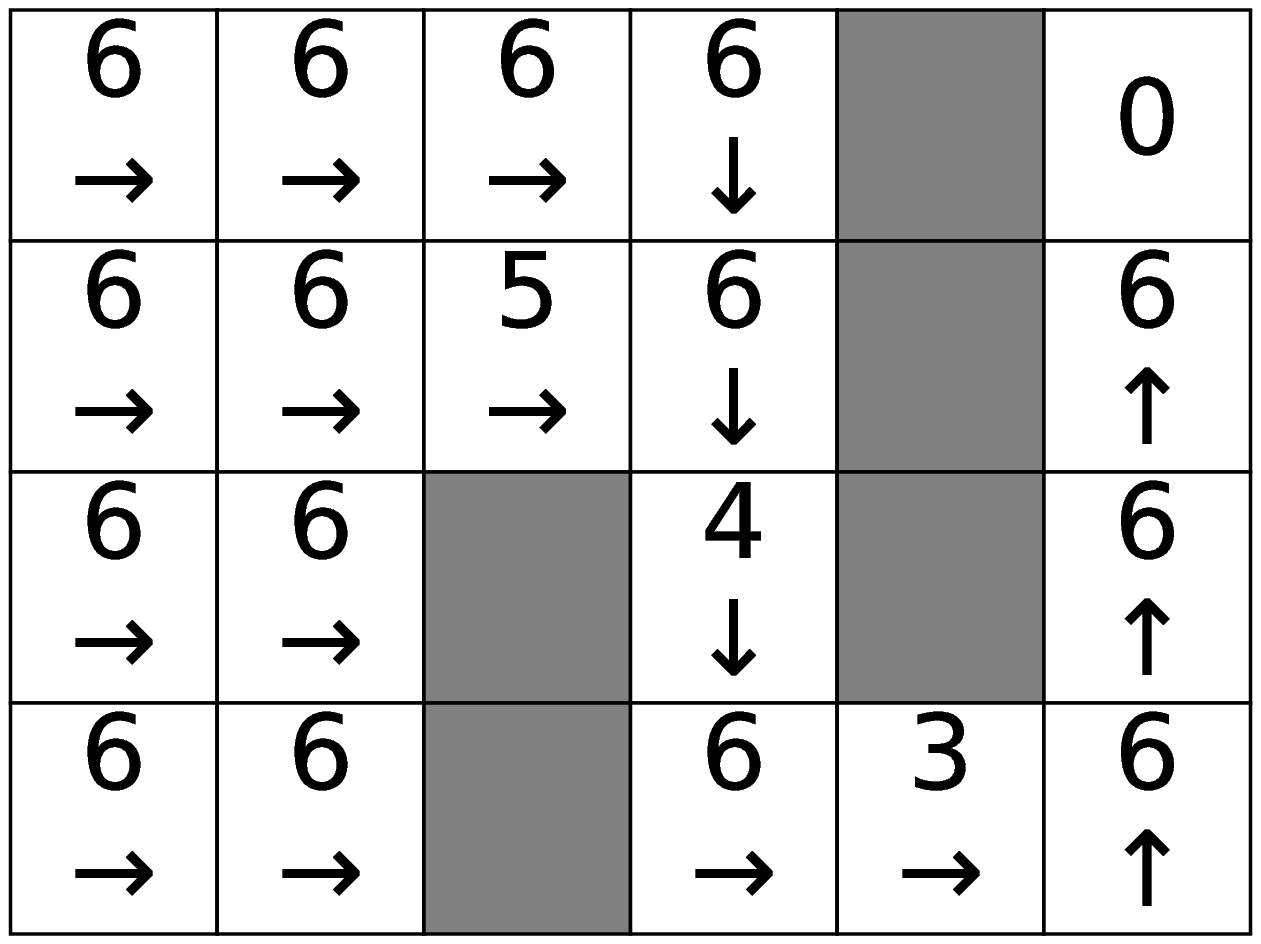}
        \caption{{\small The pre-specified level of sub-optimality penalty $\alpha = 2$.}}    
        %\label{fig:non-windyActions-d}
    \end{subfigure}
    \caption{A windy gridworld: The optimal triggering time policy $\tau_2^*(x)$ (the upper value) and the optimal control policy $\pi_2^*(x)$ (the lower pointers) for each $x\in\mathcal{X}$ under various sub-optimality requirements. (For \Cref{ProblemofHardConstraint})}. 
    \label{fig:WindyActions_Problem2}
\end{figure}

% In conclusion, in this section, we show the computation and the implementation of the two self-triggering polices that solves \Cref{ProblemOfPenalty} and \Cref{ProblemofHardConstraint} respectively. The self-triggering policy for \Cref{ProblemOfPenalty} solves the problem with a dynamic programming nature by incorporating an update penalty and the policy is optimal under the criterion \cref{Eq:CostFunPenaltyProblem}. The self-triggering policy for \Cref{ProblemOfPenalty} solves the problem with a greedy nature. The optimal time policy might not be optimal. But the policy guarantees a pre-specified level of sub-optimality. 

\section{Conclusions}
In this paper, two self-triggering policies are obtained by proposing two frameworks that convey two different philosophies. \Cref{ProblemOfPenalty} introduces a soft constraint, i.e., a update penalty that penalizes frequent use of communication resources and \Cref{ProblemofHardConstraint} applies a hard constraint on the level of sub-optimality while maximizing the triggering time to resources consumption. Both policies are shown to be effective in reducing the use of communication resources in the gridworld examples. Future endeavors can focus on developing stability guarantees of self-triggered policy for controlled Markov chain, and learning when to trigger, i.e., leveraging reinforcement learning techniques for unknown MDP models.

                                  % sure that you do not shorten the textheight too much.()

%%%%%%%%%%%%%%%%%%%%%%%%%%%%%%%%%%%%%%%%%%%%%%%%%%%%%%%%%%%%%%%%%%%%%%%%%%%%%%%%

%%%%%%%%%%%%%%%%%%%%%%%%%%%%%%%%%%%%%%%%%%%%%%%%%%%%%%%%%%%%%%%%%%%%%%%%%%%%%%%%

%%%%%%%%%%%%%%%%%%%%%%%%%%%%%%%%%%%%%%%%%%%%%%%%%%%%%%%%%%%%%%%%%%%%%%%%%%%%%%%%
\appendix
\subsection{Proof of \Cref{Theo:DPEquation}}\label{Proof:DPEquation}
\begin{proof}
We prove the theorem by constructing a consolidated Markov decision process problem. A close look at \cref{Eq:CostFunPenaltyProblem} shows that this is a discounted cost discrete-time MDP with discount factor $\beta$, Markov states and Markov actions given respectively by
{\small$$
\begin{aligned}
X_l &= (x_{t_l},\tilde{t}_l) \in \mathcal{X}\times \{0,1,2,\cdots\},\\
A_l &= (a_{t_l}, \Delta t) \in \mathcal{A}\times \mathcal{T},
\end{aligned}
$$}
where $\tilde{t}_l = t_l -l$, the state cost equal to
{\small$$
C(X_l,A_l) = \beta^{\tilde{t}_l} \left[\bar{c}(x_{t_l},a_{t_l},\Delta t) + \beta^{\Delta t}O \right],
$$}
and the skip-transition probability defined in \cref{Eq:Skip-TransitionProb}. Hence, the cost in \cref{Eq:CostFunPenaltyProblem} becomes
{\small$$
f^\mu(x) = \mathbb{E}\left[ \sum_{l=0}^\infty \beta^l C(X_l.A_l)\middle \vert X_l = (x,0), \mu \right].
$$}
The consolidated formulation can be treated as a regular Markov decision problem. Note that the Cartesian product of countable countably many polish spaces is still Polish. Hence, $\mathcal{X}\times \{0,1,2,\cdots\}$ is Polish if $\mathcal{X}$ is polish. Thus, the results (mainly
the results available to Polish spaces) can be derived from current Markov decision literature \cite{puterman2014markov}. Applying Theorem 6.2.5 and Theorem 6.2.12 of \cite{puterman2014markov}, we obtain claims in \Cref{Theo:DPEquation}.
\end{proof}

\subsection{Proof of \Cref{Lemma:NecessaryHardConstraint}}\label{Proof:NecessaryHardConstraint}
\begin{proof}
For a given $L$, let $t_{L+1}$ be the time instance of the $(L+1)$-th update. The accumulated costs before $(t_{L+1})$ can be written as 
{\begin{equation}\label{Eq:CostBeforetL}
\begin{aligned}
\small
&\mathbb{E}\left[\sum_{t=0}^{t_{L+1}-1} \beta^t c(x_t,a_t)\middle \vert x_0 =x \right]\\ 
=& \mathbb{E}\left[ \sum_{l=0}^L \beta^{t_l} \bar{c}(x_{t_l},a_{t_l},t_{l+1}-t_{l})\middle \vert x_0=x \right]\\
=&\mathbb{E}\left[\sum_{l=0}^L \beta^{t_l} \mathbb{E}\left[ \sum_{t=t_l}^{t=t_{l+1}-1} \beta^{t-t_l}c(x_t,a_{t_l})\middle\vert x_{t_l} \right] x_0 =x \right],
\end{aligned}
\end{equation}}where we use the tower property of conditional expectation to derive the first equality and the second equality follows immediately after some algebraic rearrangements. Suppose at time instance $t_l$, $l\in\mathbb{N}$, the process is at state $x_{t_l}$. Let $t_{l+1} = t_l + \tau(x_{t_l})$ and $a_{t} = \pi(x_{t_l})$ for $t=t_l,t_l +1,\dots,t_{l+1}-1$. From \cref{Eq:NecessaryHardConstraint}, we have
{
\begin{equation}\label{Eq:OngoingInquality}
\small
\mathbb{E}\left[ \sum_{t=t_l}^{t_{l+1}-1} \beta^{t-t_l} c(x_t,a_t)\middle\vert x_{t_l}\right] \leq \alpha V(x_{t_l}) -\mathbb{E}\left[ \alpha\beta^{t_{l+1}-t_l} V(x_{t_{l+1}})\middle\vert x_{t_l}\right].
\end{equation}
}
Applying \cref{Eq:OngoingInquality} into \cref{Eq:CostBeforetL} for every $l\leq L$ yields
$$
\small
\begin{aligned}
&\mathbb{E}\left[\sum_{t=0}^{t_{L+1}-1} \beta^{t} c(x_t,a_t)\middle \vert x_0 =x \right]\\ 
\leq& \mathbb{E}\left[\sum_{l=0}^L \beta^{t_l} \alpha\left( V(x_{t_l}) - \beta^{t_{l+1}-t_l} V(x_{t_{l+1}}) \right)   \middle \vert x_0=x\right] \\
=& \alpha \mathbb{E}\left[ V(x_{t_0}) - \beta^{t_{L+1}} V(x_{t_{L+1}}) \middle \vert x_0 =x\right]
\leq \alpha V(x),
\end{aligned}
$$
where we use the fact that $c:\mathcal{X}\times \mathcal{A}\rightarrow \mathbb{R}^+$ produces a non-negative $V(\cdot)$, i.e., $V(x)\geq 0,\forall x\in\mathcal{X}$. Since $L$ can be chosen arbitrarily, taking $L$ to infinity, we have $t_l\rightarrow \infty$, and by definition of $v^\mu$, $v^\mu(x) \leq \alpha V(x)$ for every $x\in\mathcal{X}$.
\end{proof}

\subsection{Proof of \Cref{Theo:Well-Defineness}}\label{Proof:Well-Defineness}
\begin{proof}
To show that there is a feasible set for problem (\ref{Eq:OptimizationHardConstraint}), it is sufficient to show that for any $x\in\mathcal{X}$, when $\Delta t_x =1$, there always exists an action $a_x\in\mathcal{A}$ such that 
$
\mathbb{E}\left[ c(x,a_x) + \alpha \beta V(x_{1})\middle\vert x_0 =x, a_0 = a_x\right] \leq \alpha V(x).
$
Let $\phi^*$ be an optimal policy of the classic MDP. Then, by Bellman equation, we have
$$
\min_{a\in\mathcal{A}}\mathbb{E}\left[ c(x,a) + \beta V(x_1)\middle\vert x_0 =x, a_0 =a \right] = V(x),
$$
where the minimum is attained at $a^* = \phi^*(x)$. That means there exists $a_x = \phi^*(x)$ such that
$$
\alpha\mathbb{E}\left[ c(x,a_x) + \beta V(x_1)\middle\vert x_0 =x, a_0 =a_x \right] = \alpha V(x).
$$
Since $c(\cdot,\cdot)$ is non-negative, we have
$$
\mathbb{E}\left[ c(x,a_x) + \alpha\beta V(x_1)\middle\vert x_0 =x, a_0 =a_x \right] = \alpha V(x).
$$
This shows that for every $x\in\mathcal{X}$, there always exists a $\Delta t_x \in\mathcal{T}$ such that we can find an action $a_x\in\mathcal{X}$ so that (\ref{Eq:NecessaryHardConstraint}) is satisfied.
\end{proof}


%\section*{ACKNOWLEDGMENT}

%The preferred spelling of the word ÒacknowledgmentÓ in America is without an ÒeÓ after the ÒgÓ. Avoid the stilted expression, ÒOne of us (R. B. G.) thanks . . .Ó  Instead, try ÒR. B. G. thanksÓ. Put sponsor acknowledgments in the unnumbered footnote on the first page.

%%%%%%%%%%%%%%%%%%%%%%%%%%%%%%%%%%%%%%%%%%%%%%%%%%%%%%%%%%%%%%%%%%%%%%%%%%%%%%%%

%References are important to the reader; therefore, each citation must be complete and correct. If at all possible, references should be commonly available publications.

\bibliography{references}

\bibliographystyle{IEEEtran}

\end{document}